\documentclass[twocolumn,pra,superscriptaddress]{revtex4-2}

\usepackage[utf8]{inputenc}
\usepackage{float}
\usepackage{placeins}
\usepackage{qcircuit}
\usepackage{braket}
\usepackage{amsmath}
\usepackage[normalem]{ulem}
\usepackage{xcolor}
\usepackage{multirow}
\usepackage{braket}
\usepackage{soul}

\usepackage{graphicx}
\usepackage{dcolumn}

\usepackage{bm}
\usepackage{bbold}

\usepackage{dsfont}

\usepackage{hyperref}
\usepackage{amsthm}
\usepackage{verbatim} 
\usepackage{circuitikz}
\usepackage{tabularx}
\usepackage{color,soul}
\usepackage{esint}
\usepackage{physics}
\usepackage{amssymb}
\usepackage{subfigure}

\usepackage{tabto}
\usepackage{listings}
\usepackage[english]{babel}
\usepackage{thmtools,thm-restate}

\newtheorem{lemma}{Lemma}
\newtheorem{proposition}{Proposition}
\newtheorem{remark}{Remark}

\usepackage{mathtools} 
\graphicspath{{Figures/}}

\usepackage{diagbox} 
\usepackage{makecell}

\newcommand{\Mod}[1]{ (\mathrm{mod}~#1)}
\newcommand{\Bin}[1]{ {\textrm{BIN}(#1)}}
\newcommand{\BinL}[1]{ {\textrm{BINL}(#1)}}

\newcommand{\BinCut}[2]{ {\textrm{BIN}_{#1}(#2)}}

\newcommand{
\MBinCut}[2]{ {\textrm{MBIN}}_{#1}(#2)}

\newcommand{
\MBin}[1]{ {\textrm{MBIN}}(#1)}

\newcolumntype{P}[1]{>{\centering\arraybackslash}p{#1}}
\newcolumntype{M}[1]{>{\centering\arraybackslash}m{#1}}

\makeatletter
\renewcommand*\env@matrix[1][*\c@MaxMatrixCols c]{%
  \hskip -\arraycolsep
  \let\@ifnextchar\new@ifnextchar
  \array{#1}}
\makeatother

\begin{document}

\preprint{APS/123-QED}
\title{High order schemes for solving partial differential equations on a quantum computer}

\author{Boris Arseniev} \affiliation{Skolkovo Institute of Science and
Technology, Moscow, Russian Federation}
\author{Dmitry Guskov}\affiliation{Skolkovo Institute of Science and
Technology, Moscow, Russian Federation}
\author{Richik Sengupta}\affiliation{Artificial Intelligence Research Institute, Moscow, Russia}\affiliation{Skolkovo Institute of Science and
Technology, Moscow, Russian Federation}
\author{Igor Zacharov}\affiliation{Skolkovo Institute of Science and
Technology, Moscow, Russian Federation}

\begin{abstract}
We explore the utilization of  higher-order discretization techniques in optimizing the gate count needed for quantum computer based solutions of partial differential equations. To accomplish this, we present an efficient approach for decomposing $d$-band diagonal matrices into Pauli strings that are grouped into mutually commuting sets.

Using numerical simulations of the one-dimensional wave equation, we show that higher-order methods can reduce the number of qubits necessary for discretization, similar to the classical case, although they do not decrease the number of Trotter steps needed to preserve solution accuracy. This result has important consequences for the practical application of quantum algorithms based on Hamiltonian evolution.
\end{abstract}

\maketitle

\section{Introduction} \label{sec:introduction}

Partial differential equations (PDEs) are essential for modeling a wide array of phenomena in physics, engineering, and various scientific fields. Specifically, within Computational Fluid Dynamics (CFD), widely used Finite Difference (FD) methods apply discretization of spatial and temporal components to derive solutions \cite{FluidDynamics}. In the specific context of modeling wave propagation, studies have demonstrated that utilizing higher-order accuracy in discretization boosts the efficiency of time-step integration, enhances the stability of the solutions \cite{CFDAcoustic, HigherOrder2000}, and effectively reduces numerical dispersion \cite{Liang}. However, this modeling comes with significant computational costs as the number of discretization points $N$ and the problem's dimensionality $D$ increasing as $~O(N^D)$.

Quantum computing offers a promising alternative to tackle these challenges. Since Feynman proposed his conjecture in 1982, extensive research has focused on demonstrating how quantum computers employing universal gates \cite{QCompUniversal} can enable more efficient algorithms for solving PDEs. These efforts particularly target the rapid growth in resource requirements as dimensionality increases \cite{QPDEsPhysRevA2023, Childs2021highprecision, Costa}.

For the wave equation Costa et al. proposed an algorithm  which exhibits improved scalability in three-dimensional space, specifically polynomial when compared to the traditional method using a linear equation solver \cite{Costa}. The algorithm reformulates the wave equation as a Schrodinger time evolution for a quantum state $\psi (t)$, where the established solution is $\psi (t) = e^{-i Ht} \psi (t=0)$. The quantum state time evolution is directed by the system Hamiltonian $H$, which inherently includes the discretization scheme.
The advantage of this approach is that the qubit count $n$ increases linearly with dimension $n D$ and grows logarithmically with respect to discretization as $n=\log_2 N$.

Solving PDEs can be approached by simulating the behavior of a quantum system, a thoroughly explored subject within the context of the quantum circuit model for universal quantum computation 
\cite{Hagan2023compositequantum, berry2007efficient,BerrySimulatingH2015, NielsenChuang}. The critical challenge lies in decomposing the Hamiltonian $H$, representing the system dynamics, into a sum of Pauli operators for implementation in quantum circuits. 

Standard decomposition methods involve matrix multiplications, with computational complexity scaling exponentially, posing significant obstacles for scalability. Recent advancements, such as a branchless algorithm leveraging Gray codes, achieve efficient decomposition of dense matrices into Pauli strings with improved speed and memory efficiency \cite{jones2024decomposing}. Similarly, the PauliComposer algorithm optimizes tensor products of Pauli matrices, accelerating Hamiltonian decomposition \cite{vidal2023paulicomposer}. These works highlight the importance of tailored strategies for efficient decomposition, particularly in dense and complex cases.

In our prior study \cite{tridiagonal}, we put forward a method to decompose tridiagonal Hamiltonians, effectively removing the requirement for matrix multiplication. This represents a notable advance, as it allows to  optimize quantum algorithms for PDEs with variable coefficients. This paper focuses on extending that method to include Hamiltonians derived from higher-order discretization schemes. The main contribution of our work is this decomposition method. Higher-order schemes, much like in classical computation, are essential to enhance the numerical precision of quantum algorithms.

Second aspect of our research is related to the application of product formulas in the context of higher-order discretization. 
The sequence of Pauli operators when applied in a quantum circuit must be ordered in time. The method of approximately ordering the product of Pauli operators is referred to as trotterization \cite{Trotter, Childs_Trotter}. The time evolution is calculated incrementally using small time intervals $\delta t$. Hence, quantum computation requires breaking up the time integration interval in $r$ steps with $r=t/\delta t$, akin to the classical time discretization. Therefore, trotterization results in a significant gate count in quantum algorithms. Our previous research shows that the number of gates to solve the 1D wave equation scales as $\tilde{O}(N^{2})$, where $\tilde{O}$ signifies the principal term \cite{tridiagonal}. Therefore, trotterization represents an obstacle to quantum advantage. 

In this study, we conduct a numerical experiment to examine if implementing higher-order spatial discretization can reduce the required number of trotterization steps.

This paper starts with review of notation and methods for approximating high-order derivatives.
Section \ref{sec:decomposition} centers on decomposition of the $d$-band matrix. 
This result is used  to build a circuit for simulating Hamiltonians, as illustrated by solutions to the wave equation in Section \ref{sec:application}. The concluding remarks are presented in Section \ref{sec:conclusions}.

\section{Notations and Definitions}
\label{sec:notation}

This paper employs the same notation as the previous study, with key notations detailed in Table \ref{tab:definitions_bit_strings} along with some enhancements. Refer to Chapter II in \cite{tridiagonal} for more details.

We adhere to the following conventions: 
\begin{enumerate}
    \item 
Exclude the tensor product symbol in Pauli strings $P$, e.g., $XYZY$ stands for $X \otimes Y \otimes Z \otimes Y$.
    
    \item 
The notation $\{P_1,\;P_2\}^{\otimes{n}}$ is the $n$-fold Cartesian product of the set $\{P_1,\;P_2\}$, with its elements interpreted as Pauli strings. For example, $\{P_1,\;P_2\}^{\otimes{2}}$ is equivalent to $\{P_1,\;P_2\} \otimes \{P_1,\;P_2\}$, which is the same as $\{P_1 \otimes P_1,\;P_1 \otimes P_2,\;P_2 \otimes P_1,\;P_2 \otimes P_2\}$. In a similar manner, we represent the product $P_k \otimes P_j$ in an abbreviated form as $P_k P_j$.
    \item 
The integer representation of the bit sequence $\mathbf{x}~=~(x_1, \dots, x_n)$ is expressed as $x=\sum_{j=1}^{n}x_j2^{n-j}$. That is the most significant bit (MSB) is on the left. This mapping is represented with function $\MBinCut{n}{x}$, defined in Table \ref{tab:definitions_bit_strings}. For example, $\MBinCut{8}{18}=00010010$ for $n=8$ and $x=18$.

This aligns with \cite{Welsh} for the wave function $\textbf{k}=\ket{k_1, \dots, k_n}$ in the computational basis, which corresponds to the representation by bit strings $\ket{\mathbf{k}}$. Although the numerical value $k$ is derived from the bit strings in MSB order, operators applied to individual qubit subspaces must follow a left-to-right sequence of bits in $k$, as highlighted in \cite{Welsh}.

The sequence of bits denoted by $\Bin{}$ is an arrangement where bits are listed from left to right, that is in least significant bit (LSB) notation. The proofs of the propositions discussed in the document are detailed in Appendix \ref{sec:appendix_proof} and make use of this LSB representation. According to this notation, $\BinCut{8}{18}=01001000$, and the relevant operators are listed in Table \ref{tab:definitions_bit_strings}.
    \item 
The function $\BinL{}$ returns the length (number of bits) of binary representation, for example $\BinL{13} = 4$.
    \item 
The notation $\norm{\cdot}$ represents the $l^2$-norm for vectors and the spectral norm for matrices, defined by the largest singular value. Note that for unitary matrices $U$, it holds that $\norm{U} = 1$.
\end{enumerate}

\begin{table}[h]
\renewcommand{\arraystretch}{1.7} 
    \centering
    \begin{tabular}{c|l}
    Notation & Definition \\
    \hline
    	$ \mathcal {P}=\{ I,X,Y,Z\} $ & Set of Pauli matrices \\
    \hline
    	$\oplus$ & XOR (addition modulo 2) \\
    \hline
    	$\{P_1,\;P_2\}^{\otimes{n}}$ & $n$-th Cartesian product \\
    \hline
	\makecell{$ x^p$ \\ where $x, p \in \mathbb{B}$ } & $x^p = x \oplus \overline{p} = x \oplus p \oplus 1$ \\
    \hline
	$ \overline{\mathbf{x}} $ & Negation, $\overline{\mathbf{x}} = (\overline{x}_1,\dots,\overline{x}_n)$ \\
    \hline
    $\mathbf{x}^\mathbf{y}$ & Exponentiation, $\mathbf{x}^\mathbf{y}=(x_1^{y_1},\dots,x_n^{y_n})$ \\
    \hline
    $\mathbf{x} \cdot \mathbf{y}$ & Inner product, $\mathbf{x} \cdot\mathbf{y} = \sum_{l=1}^n x_ly_l$ \\
    \hline
    $ Z^\mathbf{y} $ & $Z^\mathbf{y} \equiv \bigotimes_{l=1}^n Z^{y_l} = Z^{y_1}\otimes\cdots\otimes Z^{y_n}$ \\
    \hline
    $\mathbf{x} \ast \mathbf{y}$ & \makecell{Concatenation of binary strings \\ $x_1 \dots x_n \ast y_1 \dots y_m = x_1 \dots x_n y_1 \dots y_m$} \\
    \hline
    	$\textrm{BIN}$ & $\textrm{BIN}:\mathbb{N}\cup\{0\} \rightarrow \mathbb{B}^{n}$\\
    \hline
        $\BinL{x}$ & \makecell{ Length of a binary representation \\ of a number $x$.} \\
    \hline
    $\BinCut{a}{b}$ & \makecell{Binary representation of a number $b$ \\ of fixed length $a$, $a \geq \BinL{b}$ \\
    padded on the right with zeros if necessary}    \\
    \hline
    $\MBinCut{a}{b}$ & \makecell{Binary representation of a number $b$ \\ of fixed length $a$, $a \geq \BinL{b}$ \\
    padded on the left with zeros if necessary}   \\
    \end{tabular}
    \caption{Notational convention. Here $\mathbf{x} = (x_1,\dots,x_n) \in \mathbb{B}^n$, $\mathbf{y} = (y_1,\dots,y_n) \in\mathbb{B}^n$, and $X$ and $Z$ are Pauli matrices.}
    \label{tab:definitions_bit_strings}
\end{table}

\subsection{High-order central finite difference schemes}
\label{sec:High-order-schemes}

Following \cite{LeVeque2007} we consider derivative approximation of $f(x): \mathbb{R} \rightarrow \mathbb{R}$ a smooth function of one variable over an interval containing a particular point of interest $x$. 

We introduce the central difference operator as
\begin{equation}
    \hat{B}_{k} f(x) \equiv \sum_{j=-k}^{k} b_j f(x+jh),
    x\in \mathbb{R},
    \label{eq:diff-operator}
\end{equation}
where $h \in \mathbb{R}$ is some small value. We can choose coefficients $b_j \in \mathbb{R}$ in such way that $b_{j} = -b_{-j}$ for $j = 1, \dots, k$, with $b_{0} = 0$, and the operator $\hat{B}_k$ will approximate a derivative with an accuracy characterized by $O(h^{\kappa})$, where the accuracy order $\kappa = 2k$, that is
\begin{equation}
    |\hat{B}_k f(x) - f'(x)| = O(h^{2k}).
    \label{eq:approximation-order}
\end{equation}

The central difference method, achieving an accuracy of $\kappa = 2k$, is characterized by the row $(-b_{k}, \dots, -b_1, 0, b_1, \dots, b_{k})$. This leads to a practical expression for obtaining a $k$-band matrix structure:
\begin{equation*}
     \begin{split}
     \left. \frac{df(x)}{dx} \right |_{x=0} & = (-b_{k}, \dots, -b_1, 0, b_1, \dots, b_{k})
     \begin{pmatrix}
         f_{-k h} \\
         \vdots \\
         f_{-h} \\
         f_{0} \\
         f_{h} \\
         \vdots \\
         f_{k h} \\
     \end{pmatrix} \\
     & = \sum_{j = 1}^{k} b_j (f_{jh} - f_{-jh}),
     \end{split}
\end{equation*}
where $f_{jh} = f(jh)$, with the grid being uniform with a constant spacing $h \in \mathbb{R}$ (regular grid).

Using the program from \cite{Fornberg1998}, we deduce coefficients to approximate the first derivative on both regular and irregular grids with the desired accuracy.
Table \ref{tab:Approx-M1} contains coefficients for the approximations of the central difference operator (normalized with $h=1$) up to $\kappa = 2k=10$.

\begin{table}[h]
    \centering
    \def\arraystretch{1.5}
    \begin{tabular}{c|l}
         k & $(b_{-k}, \dots, b_{-1}, b_{0}, b_{1}, \dots, b_{k}) = (-b_{k}, \dots, -b_{1}, 0, b_{1}, \dots, b_{k})$ \\
         \hline
         1 & $\frac{1}{2}(-1,0,1)$ \\
         \hline
         2 & $\frac{1}{12}(1,-8,0,8,-1)$ \\
         \hline
         3 & $\frac{1}{60}(-1,9,-45,0,45,-9,1)$\\
         \hline
         4 & $\frac{1}{840}(3,-32,168,-672,0,672,-168,32,-3)$ \\
         \hline
         5 & $\frac{1}{2520}(-2,25,-150,600,-2100,0,2100,-600,150,-25,2)$
    \end{tabular}
    \caption{Coefficients for central approximation of the first derivative with an accuracy of $O(h^{2k})$ on a regular grid.}
    \label{tab:Approx-M1}
\end{table}

To compute an approximation to a derivative of the function $f(x)$ over an interval $(0,l)$, the operator $\hat{B}_k$ is applied to the set of points $f_j = f(x_j), j=1, \dots, N$ obtained after discretization of the function $f(x)$ over $N$ points. Boundary conditions are addressed in \ref{subsec:boundary_cond}; for now, assume $f(x) = 0$ when $x \notin (0,l)$. This procedure can be represented in the matrix form with a matrix $B_k$ of dimensions $N \times N$, acting on an $N$-dimensional vector $\vec{f}~\equiv~(f_1, \dots, f_N)^T$. The accuracy of this method is characterized by the discretization error, denoted as \(\epsilon_{\text{ds}}(k)\), and defined as  
\begin{equation}
\epsilon_{\text{ds}}(k) = \norm{\frac{d}{dx} \vec{f} - B_k \vec{f}},
\label{eq:ds_error}
\end{equation}
where \(\frac{d}{dx} \vec{f} = \left(f'(x_0), \dots, f'(x_N)\right)^T\).

Moreover, for each point, the accuracy is given by equation \eqref{eq:approximation-order}, taking into account that $h = \frac{l}{N-1}$ we have the following scaling for discretization error
\begin{equation}
    \epsilon_{\text{ds}}(k) = O \left( \sqrt{N} h^{2k} \right) \approx O \left( h^{2k-1/2} \right).
    \label{eq:scaling_ds}
\end{equation}

Notably, $B_k$ is a $k$-band diagonal matrix. A $k$-band matrix is characterized by having a lower and upper bandwidth of $k$, meaning there are $k$ non-zero diagonals both above and below the main diagonal, with all other elements being zero \cite{golub2013matrix}. For example $1$-band diagonal matrix is a tridiagonal matrix
\begin{equation*}
    B_1=
    \begin{pmatrix}
    b_1 & c_1 &  0 &  \dots & 0 &  0 & 0 \\
    a_1 & b_2 &  c_2 & \dots & 0 & 0 & 0   \\
    0 & a_2 & b_3 & \dots & 0 & 0 & 0   \\
    \vdots & \vdots & \vdots & \ddots & \vdots & \vdots & \vdots \\
    0 & 0 & 0 & \dots & b_{N-2} & c_{N-2} & 0  \\
    0 & 0 & 0 & \dots & a_{N-2} & b_{N-1} & c_{N-1} \\
    0 & 0 & 0 & \dots & 0 & a_{N-1} & b_{N} &  \\
    \end{pmatrix}.
\end{equation*}

The matrix representation of the central difference operator $\hat{B}_1$ (see eq. \eqref{eq:diff-operator} with $k=1$) is the matrix $B_1$ above with values $a_k = -1, b_k = 0, c_k = 1$ (see table \ref{tab:Approx-M1}) and assuming $f(x) = 0$, for $x \notin (0,l)$.

Repeated applications of $\hat{B}_k$ result in derivatives of higher-order. Typically, differential equations of second order or more can be converted into a series of first-order differential equations, making them appropriate for creating quantum algorithms focused on system dynamics. 

In particular, the wave equation discussed in Section \ref{subsec:quant_algor} is transformed into a system described by the Schrodinger equation.
The Hamiltonian for this system, up to a constant, is depicted by a symmetrized matrix:
\begin{equation}
    H =   
    \begin{pmatrix}    0 & B \\    B^\dagger & 0\\     \end{pmatrix},
    \label{eq:BBbar-example}
 \end{equation}
and the Laplacian becomes $L = B^2 = -B B^\dagger$ in this case.

\section{Multiband matrix decomposition}
\label{sec:decomposition}

The process of decomposing the multiband matrix $B_k$ into Pauli strings comprises two tasks: determining the Pauli strings that are part of the decomposition and computing their corresponding coefficients. Although theoretically one might consider the entire set of \(4^n\) Pauli strings, practically many of these do not play a role in the expansion. Hence, it is beneficial to minimize this set by eliminating strings that have zero coefficients. The following subsections focus on the reduction procedure and the representation of the resulting structure.

\subsection{Sets of Pauli strings in decomposition}

An arbitrary Pauli string can be defined as the image of the extended Pauli string operator (Walsh function) $\hat{W}:\mathbb{B}^n\times \mathbb{B}^n\rightarrow \mathcal{P}_n$ as follows:
\begin{equation}
\hat{W}(\mathbf{x},\mathbf{z}) = \imath^{\mathbf{x}\cdot\mathbf{z}} X^{\mathbf{x}}Z^{\mathbf{z}} = \bigotimes_{j=1}^{n}\imath^{x_jz_j}X^{x_j}Z^{z_j}.
\label{eq:walsh_oper}
\end{equation}
This makes it possible to encode any Pauli string $P$ with a unique pair $(\mathbf{x},\mathbf{z})$. Decomposition into Pauli basis for matrix $B \in \mathbb{C}^{2^n \cross 2^n}$ may be written as
\begin{equation}\label{eq: decomposition_app}
    B = \frac{1}{2^n}\sum_{\mathbf{x},\mathbf{z}\in\mathbb{B}^n}\beta_{\mathbf{x},\mathbf{z}}\hat{W}(\mathbf{x},\mathbf{z}).
\end{equation}

As shown earlier in Proposition 4 of \cite{tridiagonal}, the coefficients can be expressed in the form
\begin{equation}
    \beta_{\mathbf{x},\mathbf{z}} = \sum_{\mathbf{p}\in\mathbb{B}} \imath^{\mathbf{x}\cdot\mathbf{z}}(-1)^{\mathbf{z}\cdot\mathbf{p}}\cdot b_{\mathbf{p},\overline{\mathbf{p}}^\mathbf{x}},
    \label{eq:calc_weight}
\end{equation}
where $b_{\mathbf{p},\overline{\mathbf{p}}^\mathbf{x}}$ are elements of $B$.

For any $k$-diagonal in a matrix, index $p$ spans all values $ p \in \{0, \dots, 2^{n}-1-k\}$, which determine the value of $\mathbf{x}$ in the solution of the following equation:
    \begin{equation}\label{eq:main_band}
       \mathbf{p}+\mathbf{k} = \mathbf{p}\oplus \mathbf{x} \;\; 
    \end{equation}

The index $\mathbf{z}$ from the Walsh function \eqref{eq:walsh_oper} runs over $0, \dots, 2^{n}-1$ to count all expansion terms belonging to the same $\mathbf{x}$.

The equations \eqref{eq:walsh_oper} to \eqref{eq:main_band} are used for the unique identification of all Pauli strings $P=\hat{W}(\mathbf{x},\mathbf{z})$ that play a role in the matrix decomposition \cite{tridiagonal}.

We now build upon this concept by presenting explicit solutions to equations \eqref{eq:main_band} for all $k = 1, \dots, d$ in the subsequent proposition, which provides the necessary condition for the non-trivial decomposition of any arbitrary $d$-band matrix into Pauli strings:

\begin{restatable}[Decomposition of a $d$-band matrix]{proposition}{decompositionSets}\label{theor:sets} 
The only Pauli strings that can have non-zero coefficients in the decomposition of a $d$-band matrix $B \in \mathbb{C}^{N \times N}$, where $N = 2^n$, are given by $W(\mathbf{x}_{k,j}, \mathbf{z})$,  with $\mathbf{z} \in \mathbb{B}^n$ and  
\begin{equation}
    \mathbf{x}_{k,j} = \MBinCut{n-s}{2^j-1} \ast \MBinCut{s}{2^s-k} \, ,
    \label{eq:unique_solutions}
\end{equation}
where $k \in \{1, \dots, d\}$, $j \in \{1, \dots, n-s\}$, $s = \lceil \log_2(k) \rceil$, and $d \leq 2^{n-1}$.

Additionally, the main diagonal is encoded by strings  $W(\mathbf{x}_{0}, \mathbf{z})$, where
\begin{equation}
    \mathbf{x}_0 = \MBinCut{n}{0}.
    \label{eq:unique_solutions_0}
\end{equation}
\end{restatable}

All Pauli strings with non-zero coefficients in the decomposition are part of the sets labeled by $\mathbf{x}_{k,j}$ and $\mathbf{x}_0$, as described below:
\begin{equation*}
\begin{split}
 S_{k,j} &= \{ W(\mathbf{x}_{k,j},\mathbf{z}) \ | \ \mathbf{z}\in \mathbb{B}^n\},\\
 S_{0} &= \{ W(\mathbf{x}_{0},\mathbf{z}) \ | \ \mathbf{z}\in \mathbb{B}^n\}.
 \end{split}
\end{equation*}

An important use case is a \textit{symmetrized} Hermitian matrix $H$ composed of matrix $B \in \mathbb{C}^{N \times N}$ with $N=2^n$ as specified in \eqref{eq:BBbar-example}.
In corollary to proposition \ref{theor:sets} we offer the following statement:
\begin{restatable}[Decomposition of a symmetrized matrix]{corollary}{decompositionSets_symm}\label{theor:sets-BBbar} 
In the matrix decomposition of $H \in \mathbb{C}^{2N \times 2N}$ described in \eqref{eq:BBbar-example}, where $B \in \mathbb{C}^{N \times N}$ is a $d$-band matrix with $N~=~2^n$, the sets $S_{k,j}$ and $S_{0}$ consists of bit strings  which are a 1-bit extension to \eqref{eq:unique_solutions} and \eqref{eq:unique_solutions_0}:
\begin{equation}
    \mathbf{x}_{k,j} =   1\ast \MBinCut{n-s}{2^j-1}  \ast \MBinCut{s}{2^s-k},
    \label{eq:unique_solutions_BBbar}
\end{equation}
\begin{equation} \label{eq:unique_solutions_BBbar_0}
    \mathbf{x}_0 = 1\ast \MBinCut{n}{0}.
\end{equation}

\end{restatable}

The unique pair $(\mathbf{x}_{k,j},\mathbf{z})$ clearly specifies each Pauli string within the decomposition of $B$, as described by equation \eqref{eq:walsh_oper}, in a one-to-one correspondence.  
Take note of the following points:

\begin{itemize}
    \item 
    The bit strings obtained in Proposition \ref{theor:sets} (formulas \eqref{eq:unique_solutions} and \eqref{eq:unique_solutions_0}) can be used to derive the Pauli strings, which can have non-zero coefficients in the decomposition using \eqref{eq:walsh_oper}.
    \item 
    The decomposition characterization via Proposition \ref{theor:sets}, using the $(\mathbf{x}_{k,j},\mathbf{z})$ pair, allows us to obtain the appropriate Pauli gates aiding in quantum circuit design (see Section \ref{sec:best_praxis}).
    \item 
In Proposition \ref{theor:sets}, the sets $S_{k,j}$ for $d = 2^{n-1}$ encompass all $4^n$ Pauli strings. Consequently, to decompose a complete matrix, one can calculate all $(\mathbf{x},\mathbf{z})$ pairs using formula \eqref{eq:unique_solutions} with $d = 2^{n-1}$ and determine the corresponding weights.
    \item 
    This decomposition technique results in subsets of Pauli operators labeled by $\mathbf{x}_{k,j}$, which possess advantageous commuting characteristics as explored in Corollary \ref{corr:commute}.
\end{itemize}  

Within each set labeled by $\mathbf{x}_{k,j}$ and $\mathbf{x}_{0}$ the elements are generated by $\mathbf{z} \in \mathbf{B}^n$, which varies from $0,\ldots, 2^{n}-1$, therefore set cardinality is $2^{n}$, same as in previous study \cite{tridiagonal}.
  
The number of sets is determined by the following proposition:

\begin{restatable}[Number of sets in the decomposition]{proposition}{decompositionNumSets}\label{theor:number_sets} 

The total count of sets $S_{k,j}$ (including $S_{0}$) in the decomposition of a $d$-band matrix $B \in \mathbb{C}^{N \times N}$, where $N = 2^n$, as described in Proposition \ref{theor:sets} is given by 
\begin{equation}\label{eq:number_sets}
    s(d,n) = 2^{\BinL{d}}+(n-\BinL{d})d ,
\end{equation}
where $\BinL{d}$ is the binary length of $d$.
\end{restatable}

Our decomposition method, as described in Proposition \ref{theor:sets}, provides the benefit of organizing all commuting operators within each subset. This result, derived from our earlier research \cite{tridiagonal}, relies on the parity of $Y$ operators and is demonstrated in the ensuing corollary.

\begin{restatable}[Commuting subsets]{corollary}{CorollaryCommute}\label{corr:commute}
The Pauli strings that appear in the decomposition of a $d$-band matrix $B \in \mathbb{C}^{N \times N}$, $N = 2^n$, can be grouped into $2 (s(d,n) - 1)$ internally commuting subsets
 distinguished by 
 \begin{equation} 
\mathbf{x} \cdot \mathbf{z} = 0 \Mod{2} \;\;\; \text{or} \;\;\; \mathbf{x} \cdot \mathbf{z} = 1 \Mod{2}.
\label{eq;Yparity}
\end{equation} 
\end{restatable}

\subsection{Decomposition example} \label{subsec:visualization}

In this subsection, the application of Propositions \ref{theor:sets} and \ref{theor:number_sets}, together with Corollary \ref{corr:commute}, is demonstrated using the example of a $3$-band matrix $B \in \mathbb{C}^{8 \times 8}$,
\begin{equation*}
B = 
    \begin{pmatrix}
    b_{1,1} & b_{1,2} & b_{1,3} & b_{1,4} & 0 & 0 & 0 & 0 \\
    b_{2,1} & b_{2,2} & b_{2,3} & b_{2,4} & b_{2,5} & 0 & 0 & 0 \\
    b_{3,1} & b_{3,2} & b_{3,3} & b_{3,4} & b_{3,5} & b_{3,6} & 0 & 0 \\
    b_{4,1} &  b_{4,2} & b_{4,3} & b_{4,4} & b_{4,5} & b_{4,6} & b_{4,7} & 0 \\
    0 & b_{5,2} &  b_{5,3} & b_{5,4} & b_{5,5} & b_{5,6} & b_{5,7} & b_{5.8} \\
    0 & 0 & b_{6,3} &  b_{6,4} & b_{6,5} & b_{6,6} & b_{6,7} & b_{6,8} \\
    0 & 0 & 0 & b_{7,4} & b_{7,5} & b_{7,6} & b_{7,7} & b_{7,8} \\
    0 & 0 & 0 & 0 & b_{8,5} & b_{8,6} & b_{8,7} & b_{8,8} \\
    \end{pmatrix}.
\end{equation*}

First, we apply Proposition \ref{theor:sets} to determine all strings $\mathbf{x}_{k,j}$ as defined by equations \eqref{eq:unique_solutions} and \eqref{eq:unique_solutions_0}. A (pseudo-) C code to calculate equation \eqref{eq:unique_solutions} for all $\mathbf{x}$ is given in Appendix \ref{list:unique_solutions}. 

Pauli strings are constructed by considering all possible strings $\mathbf{z} \in \mathbb{B}^n$. From equation \eqref{eq:walsh_oper}, each $(\mathbf{x}_{k,j}, \mathbf{z})$ pair translates into an $n$-length Pauli string, where each bit pair $(x_l, z_l)$, for $l = 1, \dots, n$, yields a Pauli matrix: $I \equiv (0,0)$, $Z \equiv (0,1)$, $X \equiv (1,0)$, or $Y \equiv (1,1)$. Program listing \ref{list:walsh_oper} demonstrates function $W$ that generates these strings. It returns the sign $(-1)^{\mathbf{x}\cdot\mathbf{z}}$ as stated in \eqref{eq:walsh_oper}.

The outcome is presented in Table \ref{tab:subsets_d3}, consisting of a total of $7$ sets, in accordance with Proposition \ref{theor:number_sets}. Note that the function in Appendix \ref{list:unique_solutions} also returns the number of sets. Each set contains $2^n$ members. 

Following Corollary \ref{corr:commute}, each set can be further subdivided based on the parity of the number of $Y$ operators, leading to the identification of internally commuting subsets. These subsets are highlighted in blue and red, representing even and odd parity, respectively. Notably, if $B$ is a symmetric matrix, the ``red subset'' does not appear.

\begin{table}[h]
    \centering\small
    \begin{tabular}{c|c|c}
         \shortstack{Set $\mathbf{x}_{k,j}$} & 
         $(\mathbf{z})$ set member & Pauli strings \\
         \hline 
         $\mathbf{x}_{0}=000$ &  
         \shortstack{\textcolor{blue}{000, 001, 010, 011} \\
                     \textcolor{blue}{100, 101, 110, 111}} &
         \shortstack{\textcolor{blue}{III,IIZ,IZI,IZZ,}\\ 
                     \textcolor{blue}{ZII,ZIZ,ZZI,ZZZ}}  \\
         \hline 
         $\mathbf{x}_{1,1} =001$ &  
         \shortstack{\textcolor{blue}{000, 010, 100, 110} \\ 
                     \textcolor{red} {001, 011, 101, 111}} &
         \shortstack{\textcolor{blue}{IIX,IZX,ZIX,ZZX} \\                      \textcolor{red} {IIY,IZY,ZIY,ZZY}} \\
         \hline 
         $\mathbf{x}_{1,2} =011$ &  
         \shortstack{\textcolor{blue}{000, 100, 011, 111} \\
                     \textcolor{red} {001, 010, 101, 110}} &
         \shortstack{\textcolor{blue}{IXX,ZXX,IYY,ZYY} \\                      \textcolor{red} {IXY,IYX,ZXY,ZYX}}      \\
         \hline 

         $\mathbf{x}_{1,3}=111$ &  
         \shortstack{\textcolor{blue}{000, 011, 101, 110 } \\
                     \textcolor{red} {001, 010, 100, 111 }} &
         \shortstack{\textcolor{blue}{XXX,XYY,YXY,YYX} \\                      \textcolor{red} {XXY,XYX,YXX,YYY}} \\
         \hline 
         $\mathbf{x}_{2,1}=010$ &  
         \shortstack{\textcolor{blue}{000, 001, 100, 101} \\
                     \textcolor{red} {010, 011, 110, 111}} &
         \shortstack{\textcolor{blue}{IXI,IXZ,ZXI,ZXZ} \\                      \textcolor{red} {IYI,IYZ,ZYI,ZYZ}}     \\
         \hline 
         $\mathbf{x}_{2,2}=110$ &  
         \shortstack{\textcolor{blue}{000, 001, 110, 111} \\
                     \textcolor{red} {100, 101, 010, 011}}  &
         \shortstack{ \textcolor{blue}{XXI,XXZ,YYI,YYZ} \\                      \textcolor{red} {YXI,YXZ,XYI,XYZ}}     \\
         \hline        
         $\mathbf{x}_{3,1}=101$ &  
         \shortstack{ \textcolor{blue}{000, 010, 101, 111} \\
                      \textcolor{red} {001, 100, 011, 110}} &
         \shortstack{ \textcolor{blue}{XIX,XZX,YIY,YZY} \\                      \textcolor{red} {XIY,YIX,XZY,YZX}} \\  

    \end{tabular}
    \caption{Pauli strings generated from the matrix decomposition method (Proposition \ref{theor:sets}) labeled by $(\mathbf{x}_{k,j},\mathbf{z})$ for 3-band $8\times 8$ matrix $B$. Internally commuting subsets for each $\mathbf{x}$ are characterized by \eqref{eq;Yparity} indicated by blue and red.}
    \label{tab:subsets_d3}
\end{table}

We illustrate how Pauli strings enter the decomposition with a Figure \ref{fig:visualization_8x8} where we denote each set corresponding to $\mathbf{x}_{k,j}$ as $S_{k,j}$. Each set $S_{k,j}$ is listed at the place of matrix elements that contribute to that set. The colors match the diagonals $d$ in that contribution. For the $3$-band matrix in our example, the green, orange, and yellow positions are utilized.
 
\begin{figure}[h] \centering
\includegraphics[scale=0.9]{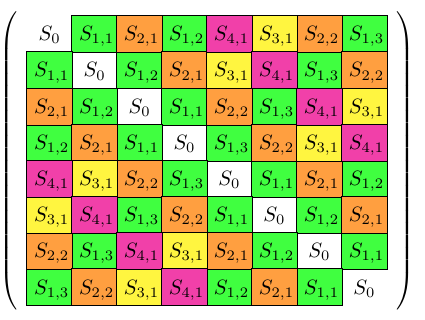}
\caption{The figure illustrates the decomposition discussed in Proposition \ref{theor:sets}. Rather than a single element, all potential Pauli strings that could contribute to the decomposition are represented as a set $S_{k,j}$. Each set corresponds to $\mathbf{x}_{k,j}$ as defined in \eqref{eq:unique_solutions}. Sets corresponding to the same bandwidth $d$ are depicted in the same color.} 
\label{fig:visualization_8x8}
\end{figure}

The weight of each Pauli string in the decomposition can be computed using equation \eqref{eq:calc_weight} with complexity $O(N\log_2(N))$. Importantly, the weights are determined uniquely by the pair $(\mathbf{x},\mathbf{z})$, eliminating the need to multiply the matrices associated with the Pauli string itself. Appendix \ref{list:calc_weight} features a (pseudo-) C code demonstrating this computation.

\subsection{Practical approach to matrix decomposition} \label{sec:best_praxis}

In our research, the $(\mathbf{x}_{k,j},\mathbf{z})$ pair is utilized to populate the data structures required for the simultaneous diagonalization of internally commuting Pauli strings. The complete data structure for it in Tableau representation as detailed in \cite{Kawase} and \cite{Berg} using a structure with cardinality $m \times 2n$ bits. In this context, $m$ represents the quantity of members within the set that can be simultaneously diagonalized, whereas $n$ denotes both the number of qubits and the length of the associated Pauli string. We note that for each $\mathbf{x}$ we fill one full structure where each row corresponds to the mapping $(\mathbf{x},\mathbf{z_k})=(x_{1},\dots,x_{n}|z_{k,1},\dots,z_{k,n})$  
\begin{equation*}
    (\mathbf{x},\mathbf{z_k}),\; k=[1,m] \rightarrow \left[ \begin{matrix}
        x_{1} & \cdots & x_{n} &|& z_{1,1} & \cdots & z_{1,n} \\
         \vdots & \cdots &  \vdots &|&  \vdots & \cdots & \vdots \\
        x_{1} & \cdots & x_{n} &|& z_{m,1} & \cdots & z_{m,n} 
    \end{matrix}\right]
\end{equation*}

We systematically apply Clifford transformations composed of $H, S, CX, CZ$ gates that can simultaneously diagonalize all Pauli strings to the Tableau. Leveraging the methodology described in \cite{Kawase}, we finalize the process by clearing the left side of the Tableau, retaining only the Z-block with $z$-values. The identified sequence of Clifford gates forms the diagonalizing transformation which we will later denote $R_{k,\gamma}$ (section \ref{sec:circuit_complexity}). The Z-block corresponding to the diagonalized operators $\Lambda_{k,\gamma}$ relates to $R_z(\theta)$ rotations as described in \cite{Welsh}, where $\theta$ aligns proportionally with decomposition weights. Since the process is independent of specific weights, rotations are parametrized. The program writes out the resultant circuits in QASM-2 format \cite{qasm-2}, with parametrization via $\theta$. This approach circumvents explicit Pauli matrices generation and Pauli matrix multiplication.

\section{High-order scheme for the one-dimensional wave equation} \label{sec:application}

In this section, we investigate how wave equation with a variable speed profile $c(x)$ (a non-negative continuous function) and Dirichlet boundary conditions can be incorporated into a quantum algorithm using high-order central approximations. Consider,

\begin{equation}
\begin{split}
    &\frac{d^2}{dt^2} u (t,x) = \frac{d}{dx} \left( c^2(x) \frac{d}{dx} u (t,x) \right), \\
    &u(0,x) = u_0(x), \\
    &\frac{d}{dt} u(0,x) = 0, \\
    &u(t,0) = u(t,l) = 0.
\label{eq:1d_wave-q}
\end{split}
\end{equation}

Here, the solution $u(t,x) \in C^2((0, \infty) \times I)$ where $I=(0,l).$ 
When $u_0(x) = \sin(\pi \frac{x}{l})$ and  $c(x)$ is constant, the solution is expressed as
\begin{equation}
u_{sw}(x,t) = \sin(\pi \frac{x}{l}) \cos(\pi \frac{c t}{l}).
\label{eq:exact_wave_sol}
\end{equation}
This standing wave solution will be used as reference to compare with the numerical results in section \ref{subsec:numerical}.

\subsection{The Quantum algorithm}\label{subsec:quant_algor}
Consider the right-hand side of the wave equation
\begin{equation}
\frac{d}{dx} \left( q(x) \frac{d}{dx} u (t,x) \right) \approx  \hat{B}_{k}  \left( q(x) \hat{B}_{k} u (t,x)\right),
\end{equation}
with $q(x) = c^2(x)$, the parameter $c(x)$ is the wave propagation speed. The operator $\hat{B}_{k}$ approximates the first derivative using the central finite difference method with an accuracy of order $\kappa = 2k$, as detailed in Section \ref{sec:High-order-schemes}. It is important to note that applying the chain rule in this context is not recommended, as the expression in brackets represents the total mass flow across the boundary from a physical perspective \cite{Langtangen}.

The function $q(x)$, when discretized, is represented by a diagonal matrix $I_q$, with its values located along the diagonal. The matrix $B_k$ is introduced as the matrix representation of the operator $\hat{B}_{k}$, with the boundary conditions discussed in a subsequent subsection. Given the relationships $I_q = I_c I_c$ and $B_k^T = -B_k$, the right-hand side of the wave equation can be written in matrix form as:
\begin{equation}
B_k I_q B_k = -B_k I_c (B_k I_c)^T = -B_k(c) B_k^T(c) \equiv L_k(c),
\label{eq:laplacian_wave}
\end{equation}
where $B_k(c) \equiv B_k I_c$. This formulation allows the incorporation of variable speed coefficients into the matrix $B_k$ by scaling each column $j$ of $(B_k)_{i,j}$ by the corresponding discretized speed value $c_j$. Therefore, by using this approach, non-constant coefficients can be easily incorporated into the finite difference method. This allows for a more flexible treatment of problems involving spatially varying parameters.

Following Costa et al.~\cite{Costa}, instead of directly analyzing the original discretized wave equation, we can consider the Schrodinger equation. The key idea behind this transformation is that the right-hand side of the wave equation is replaced by the operator $-\frac{1}{h^2} B_k B_k^T$ (here we write step size $h$ explicitly). This substitution leads to the formulation of the Schrodinger equation with the Hamiltonian:
\begin{equation}
    H_k =\frac{1}{h}
    \begin{pmatrix}
    0 & B_k\\
    B_k^T & 0\\
    \end{pmatrix}.
    \label{eq:hamil_wave}
\end{equation}
That is the two-component quantum state $\vec{\psi} = (\vec{\phi}_V, \vec{\phi}_E)^T$ evolves according to:
\begin{equation}
\frac{d}{dt}
    \begin{pmatrix}
        \vec{\phi}_V\\ \vec{\phi}_E 
    \end{pmatrix}
    = -i H_k
    \begin{pmatrix}
        \vec{\phi}_V \\ \vec{\phi}_E
    \end{pmatrix}
    \label{eq:shrod_eq},
\end{equation}
where $\vec{\phi}_V$ is a solution to considered wave equation, and $h$ represents the discretization step.

Therefore, according to \eqref{eq:laplacian_wave}, if we substitute the Hamiltonian
\begin{equation}
    H_k(c) =\frac{1}{h}
    \begin{pmatrix}
    0 & B_k(c)\\
    B_k^T(c) & 0\\
    \end{pmatrix}
    \label{eq:hamil_wave_mod}
\end{equation}
in \eqref{eq:shrod_eq}, the upper component $\vec{\phi}_V$ of quantum state $\vec{\psi}$ will be a solution of the wave equation \eqref{eq:1d_wave-q}.

We also note, that if the size of matrix $B_k(c)$ is $N \times N$, with $N=2^n$, then the size of $H_k(c)$ is $2^{n+1} \times 2^{n+1}$.

\subsubsection{Dirichlet boundary conditions}\label{subsec:boundary_cond}

We consider the application of Dirichlet boundary conditions to the problem defined in \eqref{eq:1d_wave-q}. The derivative operator, formulated in matrix form (as outlined in Section \ref{sec:High-order-schemes}), encounters limitations in accurately approximating derivative values near the boundaries. This discrepancy arises because the elements within the vector resulting from this operator do not yield desired approximations for the function's derivatives at these boundary points. Specifically, as demonstrated in Section \ref{sec:High-order-schemes}, this limitation is evident when examining the first and last rows of the matrix $B_1$. This approach works only under the assumption that the function of interest is equal to zero outside the domain boundaries, as discussed in detail by Costa et al. \cite{Costa}. 
If that is not the case, the boundary conditions must be accommodated within the Pauli decomposition and the resulting boundary scheme should be no more than an order lower than those of the interior scheme to maintain scheme's accuracy \cite{Gustafsson1975}.

In this study, we employ the following methodology (more details in Appendix \ref{sec:appendix_boundary}). Given the boundary conditions $u(t,0) = u(t,l) = 0$, it follows that $\frac{d^2}{dt^2} u(t, 0) = \frac{d^2}{dt^2} u(t, l) = 0 $. This condition necessitates that the solution $u(t, x)$ be continued anti-symmetrically at the boundary when employing the central difference operator $\hat{B}_k$ (as specified in equation \eqref{eq:diff-operator}), meaning $u(t, -x) = -u(t, x)$. Meanwhile, the wave speed $c(x)$ extends symmetrically, ensuring $c(-x) = c(x)$.

In general case  $(-b_{k}, \dots, -b_1, 0, b_1, \dots, b_{k})$ the upper left corner of matrix $B_k(c)$ with zero boundary conditions and incorporated speed profile is given by

\begin{figure}[!ht]
\centering
\includegraphics[scale=0.66]{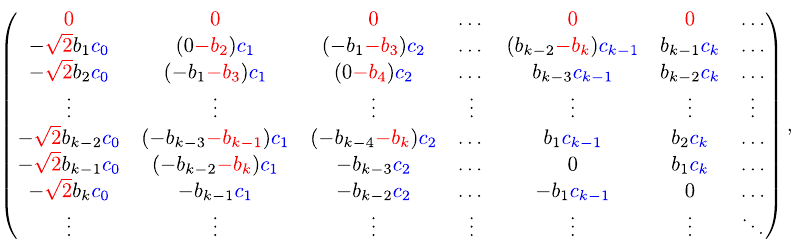}
\end{figure}
where all changes due to boundary conditions to the original matrix are shown with red and the incorporation of the speed profile $c_{j} \equiv c(x_j)$ is shown in blue. The bottom right corner mirrors the same behavior except for a reflection about the anti-diagonal.

This way of adjusting the boundary conditions has been applied for all the numerical results in our study.

\subsubsection{Time evolution of the Hamiltonian} \label{ses:Time_evolution}

To execute Hamiltonian evolution on a quantum computer, it must be decomposed into fundamental operations. This can be accomplished by expressing the Hamiltonian as a sum of terms that are easy to implement: $H_k = \sum_{\gamma=1}^{\Gamma_k} H_{k, \gamma}$ \cite{berry2007efficient}. Following earlier research \cite{tridiagonal}, $H_{k, \gamma}$ is composed of mutually commuting Pauli strings, with $\Gamma_k = s(k,n)$, where $B_k \in \mathbb{R}^{N\times N}$ and $N=2^n$. The Trotter formula, incorporating various accuracy levels, is used to address the non-commutative nature of the $H_{k, \gamma}$ operators \cite{Childs_Trotter}.

Trotter formulas of order $1$ and $2$ and higher even orders $p=2k, k=2,3,4,\dots$ are expressed as
\begin{align}
    S_{1}(t) &= e^{-itH_{\Gamma}} \dots e^{-itH_{1}}, \\
    S_{2}(t) &= e^{\frac{-it}{2}H_{1}}  \dots  e^{\frac{-it}{2}H_{\Gamma}} e^{\frac{-it}{2}H_{\Gamma}} \dots e^{\frac{-it}{2}H_{1}}, \\
    S_{2k}(t) &= S_{2k-2}^2(s_k t) S_{2k -2}((1-4s_k) t) S_{2k-1}^2(s_k t),
    \label{eq:Trotter_formula-p}
\end{align}
where $s_k = 1/(4-4^{1/(2k-1)})$. The Trotter error $\epsilon_{\text{tr}}$ is defined as
\begin{equation}
    \epsilon_{\text{tr}}(t,r) = \norm{e^{-iH_{k}t} - S_{p}^r(t/r)}, \label{eq:errors_trotter}
\end{equation}
and is influenced by the number of Trotterization steps $r$, evolution time $t$, the order of the Trotter formula $p$, and the Hamiltonian $H_k$.
According to \cite{berry2007efficient}, the scaling for the Trotter error is
\begin{equation}
    \epsilon_{\text{tr}}(t, r) = O\left(\frac{\left(2 \Gamma_k 5^{\lfloor p/2 \rfloor - 1} \norm{H_k}t \right)^{p+1}}{r^{p}}\right).
    \label{eq:scaling_tr}
\end{equation}
Each Trotterization step progresses the integration by $\delta t=\frac{t}{r}$, and to enhance computational precision, $\delta t$ should be diminished by increasing $r$. We examine whether a more precise discretization scheme can affect the Trotter time step.

To that end we assess the precision of our method by employing these metrics: the discretization approximation error $\epsilon_{\text{ds}}$, outlined in equation \eqref{eq:ds_error}; the Trotter error $\epsilon_{\text{tr}}$ from Equation \eqref{eq:errors_trotter}; and the numerical solution error $\epsilon_{\text{ns}}$, which encapsulates the total inaccuracy from all approximation stages:
\begin{align}
    \epsilon_{\text{ns}}(t) &= \norm{\vec{u}(t) - \vec{\phi}_V(t)},
   \label{eq:errors}
\end{align}
where, $\vec{u}(t)$ represents the discretized exact solution (for example, in case of standing wave it is given by \eqref{eq:exact_wave_sol}), and $\vec{\phi}_V(t)$ is the upper portion of the state vector $\vec{\psi}(t)$ (see \eqref{eq:shrod_eq}). We note that, since the initial condition for the quantum algorithm is normalized, the exact solution must be adjusted accordingly to reflect this normalization.

\subsection{Numerical results} \label{subsec:numerical}

This section presents the numerical experiments designed to assess the performance and accuracy of the proposed quantum algorithm for the solution of the wave equation based on the system dynamics.  During these experiments, the wave equation \eqref{eq:1d_wave-q} was evaluated with parameters set to $l=5$, $t=1$, and $c=1$. In alignment with \cite{Costa}, and chosen initial conditions in \eqref{eq:1d_wave-q} in all numerical experiments in this section we set initial condition of the Schrodinger equation \eqref{eq:shrod_eq} as 
\begin{equation}
\vec{\psi}(t=0) = 
    \begin{pmatrix}
        \vec{\phi}_V (t=0) \\ \vec{\phi}_E (t=0)
    \end{pmatrix} =
        \begin{pmatrix}
        \vec{u_0}/\norm{\vec{u_0}} \\ \vec{0},
    \end{pmatrix}
    \label{eq:quantum_init}
\end{equation}
this way $\norm{\vec{\psi}(t=0)}=1$. As mentioned in previous section, the solution to the wave equation will be given by first component of $\vec{\psi}(t)$, that is $\vec{\phi}_V (t)$, or, more precisely
\begin{equation}
    \vec{\phi}_V (t) = \frac{\vec{u}_{sw}(t)}{\norm{\vec{u_0}}}.
    \label{eq:quantum_solution}
\end{equation}
Thus, in what follows, we compare the results obtained in numerical simulations with equation \eqref{eq:quantum_solution}. We also recall that the number of discretization points in this section is given by $N = 2^n$ with $n+1$ representing the total number of qubits necessary for the Hamiltonian simulation.

The first numerical experiment was aimed to assess the effect of discretization error and correctness of implementation of boundary conditions for high order schemes. To do so we calculated  $e^{-i H_k t} \vec{\psi} (t=0)$ directly, with progressively finer discretization and compared it to the exact solution \eqref{eq:quantum_solution}. Figure \ref{fig:standing_comparison_error_qubits} shows that we achieve high computational accuracy with a modest number of qubits. Therefore, we performed numerical experiments with up to $8$ qubits. 

\begin{figure}[ht]
\centering
\includegraphics[width=\linewidth]{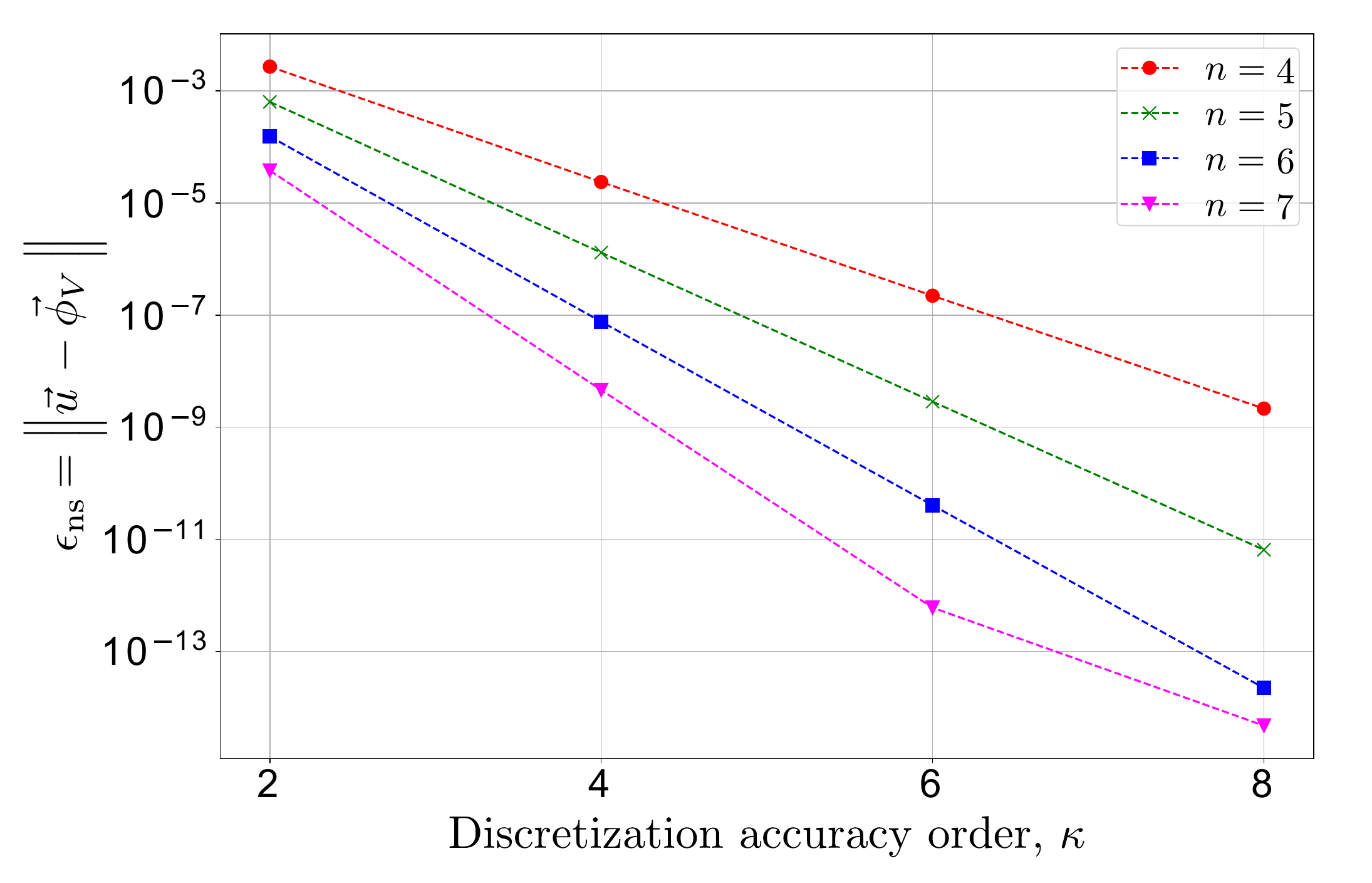}
\caption{Comparison of $\vec{\phi}_V (t)$ - solution obtained as $e^{-i H_k t} \vec{\psi} (t=0)$ with $\vec{u} (t) = \vec{u}_{sw}(t)/\norm{\vec{u_0}}$ an analytical standing wave solution at $t=1$, with $l=5$ and $c=1$ for different accuracy orders $\kappa=2k$ and different number of qubits for discretization $n$ (the total number of qubits is given by $n+1$).} 
\label{fig:standing_comparison_error_qubits}
\end{figure}

Moreover, Figure \ref{fig:standing_comparison_error_qubits} illustrates that increasing the approximation order leads to fewer qubits needed, indicating (a) the Hamiltonian in \eqref{eq:hamil_wave} efficiently represents higher-order discretization, and (b) verifies the accurate application of boundary conditions outlined in Section \ref{subsec:boundary_cond}.

\subsubsection{Results for the quantum algorithm}

The quantum algorithm for simulating Hamiltonian evolution requires executing $r$ Trotterization steps, with the number of steps affecting the precision of the solution. We employed binary search to find the minimal number of steps $r$ needed for a specific accuracy based on the chosen discretization context $(n,\kappa)$. To enable comparison of numerical results, the Trotterization order $p$ was fixed at $p=2$. This is a pragmatic decision, as a lower $p$ results in too many $r$ steps, reducing $t/r$ to below machine precision, while with $p \geq 4$, the number of steps is too small for a detailed comparison of results.

After determining the Trotterization steps $r$ for each precision goal in the discretization context $(n,\kappa)$, Figure \ref{fig:fixed_trotter_error} illustrates our key results, examining the effects of different qubit counts and discretization methods on accuracy.

\begin{figure}[ht]
\centering
\includegraphics[width=\linewidth]{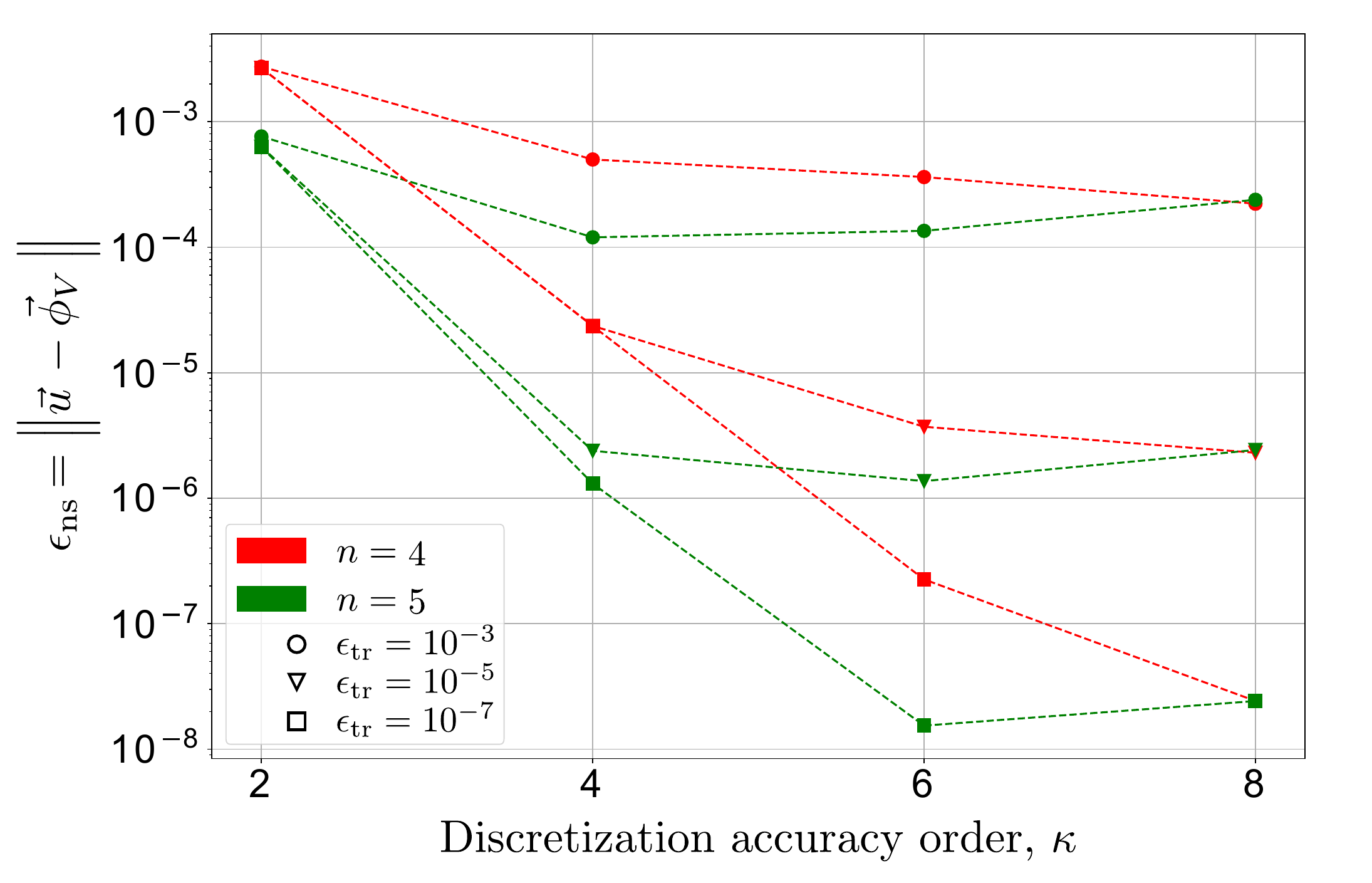}
\caption{Comparison of $\vec{\phi}_V (t)$ - solution obtained with proposed quantum algorithm with $\vec{u} (t) = \vec{u}_{sw}(t)/\norm{\vec{u_0}}$ with fixed Trotter errors of $10^{-3}$, $10^{-5}$, and $10^{-7}$, represented by different marker shapes.  Results for $4$ and $5$ qubits are given with red and green respectively (the total number of qubits is given by $n+1$). The Trotter formula is set to order $p=2$; the wave equation parameters $l=5$, $t=1$, and $c=1$.} 
\label{fig:fixed_trotter_error}
\end{figure}

Figure \ref{fig:fixed_trotter_error} demonstrates the following points:
\begin{itemize}
    \item 
The discretization precision shown in Figure 1 is attainable with an adequate number of Trotterization steps. 
For instance, using 5 qubits and a 6th-order accuracy scheme results in an overall error of approximately $10^{-8}$, paralleling the expectations illustrated in figure \ref{fig:standing_comparison_error_qubits}.
    \item 
 By employing a higher-order discretization scheme, enhanced solution accuracy is achievable using the same qubit count. For instance, moving from a 4th to a 6th order scheme with 5 qubits improves accuracy, provided it is not constrained by the imposed Trotterization's accuracy level.
 \end{itemize}
Together, Figure \ref{fig:standing_comparison_error_qubits} and Figure \ref{fig:fixed_trotter_error} illustrate that increasing qubits for discretization or employing higher-order schemes improves accuracy, provided enough Trotterization steps are executed to surpass the plateau shown in Figure \ref{fig:fixed_trotter_error}. Thus, this is our final conclusion:
 \begin{itemize}
    \item 
Trotterization accuracy serves as the main limitation for the precision of solutions due to the necessity of providing sufficient number of Trotterization steps, which notably impacts the gate complexity (see equation \eqref{eq:resulted_gates_p-2}). Even with the use of sophisticated discretization techniques and an increase in qubit numbers, the attained precision will be confined by the specified number of Trotterization steps. 
\end{itemize}

Examining the final finding above, we suggest that the overall numerical error results from the addition of the Trotter error $\epsilon_{\text{tr}}$ and the discretization error $\epsilon_{\text{ds}}$, as outlined in the following proposition.

\begin{restatable}[Error for solving wave equation]{proposition}{errorInter}\label{theor:error_inter} 

The bound for numerical solution error $\epsilon_{\text{ns}}(t)=\norm{\vec{u}(t)-\vec{\phi}_V(t)}$, where $\vec{u}(t)$ is the exact solution of the wave equation and $\vec{\phi}_V(t)$ is the numerical solution obtained 
with $2k$ discretization order scheme and $r$ Trotter steps, can be approximated as
\begin{equation}
\begin{split}
\epsilon_{\text{ns}} (t) &\leq \epsilon_{\text{tr}}(t, r) + tc_{max}\epsilon_{ds}(k),
    \label{eq:error_interconect}
\end{split}
\end{equation}
with $c_{\max} = \max_{0\leq x \leq l} c(x)$, where $c(x)$ and $l$ are from \eqref{eq:1d_wave-q}, $\epsilon_{\text{tr}}(t, r)$ represents the Trotter error  \eqref{eq:errors_trotter}, and $\epsilon_{\text{ds}}(k)$ represents the discretization error  \eqref{eq:ds_error}. 
\end{restatable}

Consequently, in order to evaluate how discretization precision influences the results and considering equations \eqref{eq:scaling_tr} and \eqref{eq:scaling_ds}, the scaling of the error in the numerical solution can be expressed as
\begin{equation*}
    \epsilon_{\text{ns}} (t) = O\left(\frac{\left(2 \Gamma 5^{\lfloor p/2 \rfloor - 1} \norm{H_k}t \right)^{p+1}}{r^{p}} +  t c_{\max} h^{2k-1/2} \right),
\end{equation*}
which indicates that for the fixed error the number of Trotter steps grows polynomially as the discretization accuracy improves. This is illustrated in figure \ref{fig:trotter_power} for a given $\epsilon_{ns} = 10^{-5}$.

To determine the minimum number of Trotter steps required to achieve the desired precision, we used a binary search algorithm. Figure \ref{fig:trotter_power} illustrates the results, indicating the relationship between the necessary Trotter steps and both the number of qubits for discretization $n$ and the discretization scheme's accuracy $\kappa$. It is important to note that for $\kappa = 4$, the graph starts at the 5th qubit, as achieving an accuracy of $10^{-5}$ with fewer qubits is infeasible, as demonstrated in Figure \ref{fig:standing_comparison_error_qubits}.

\begin{figure}[H]
\centering
\includegraphics[width=\linewidth]{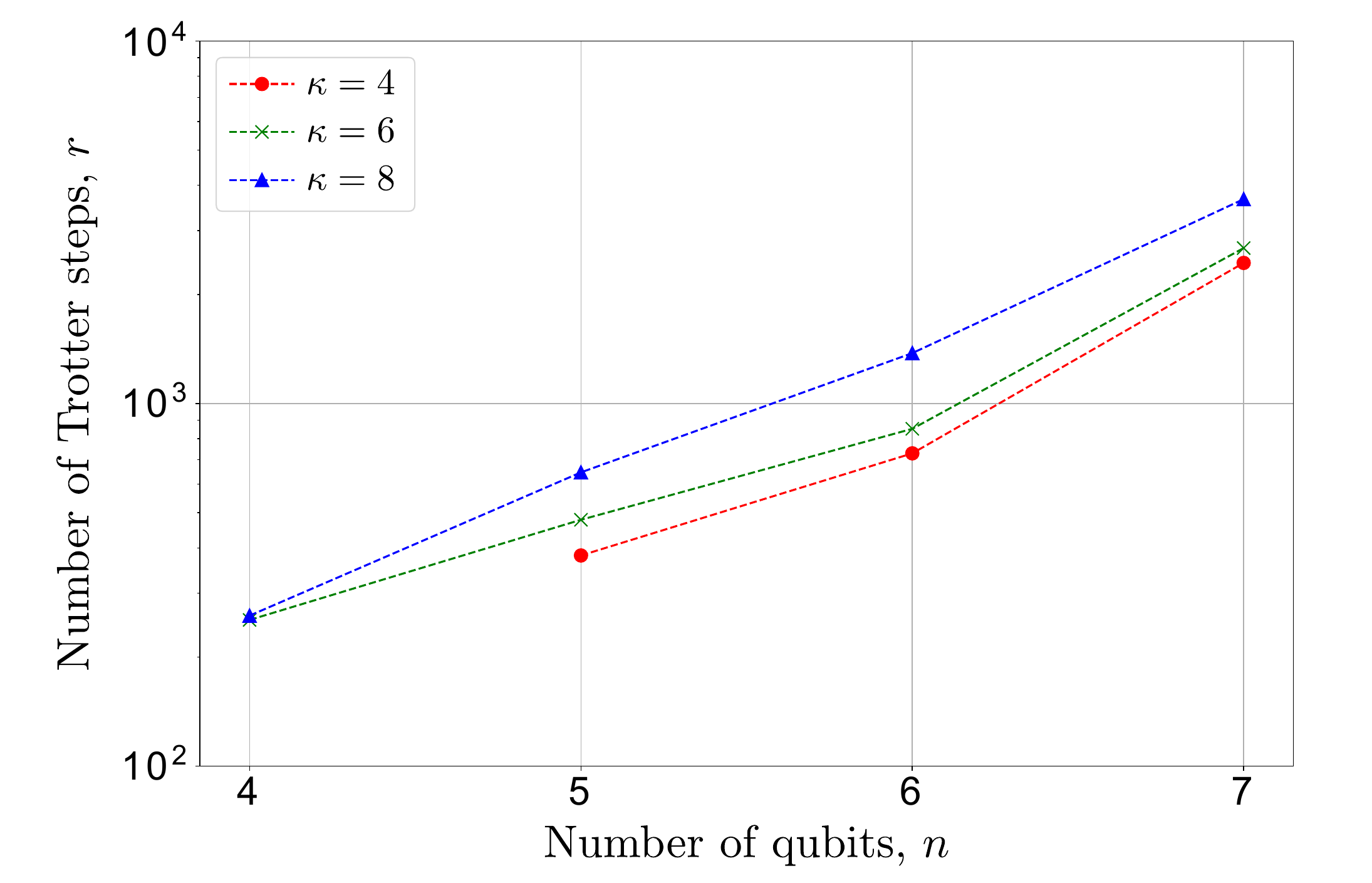}
\caption{The relationship between the number of Trotter steps $r$ and the number of qubits for a discretization $n$ (total qubits $n+1$), using various discretization orders $\kappa$, necessary to reach an error of $\epsilon_{\text{ns}} = 10^{-5}$. The Trotter formula is utilized at order $p=2$, with wave equation parameters set to $l=5$, $t=1$, and $c=1$.} 
\label{fig:trotter_power}
\end{figure}

In numerous practical scenarios, reaching a specified error level, symbolized as $\epsilon_{\text{set}}(t, k, r)$, is essential. Proposition \ref{theor:error_inter} offers a structure to guarantee that the algorithm stays within this error limit. By selecting 
\begin{equation*}
    \epsilon_{\text{set}}(t, k, r) = \epsilon_{\text{tr}}(t, r) + tc_{max}\epsilon_{ds}(k)
\end{equation*} 
we make certain that $\epsilon_{\text{ns}}(t) \leq \epsilon_{\text{set}}(t, k, r)$.

Moreover, it is clear that $\epsilon_{\text{tr}}(t, r)$ must not surpass $\epsilon_{\text{set}}(t, k, r)$ (in other words, the number of Trotter steps $r$ cannot be decreased anymore without breaching the error limit). Hence, the least number of Trotter steps is defined when $\epsilon_{\text{tr}}(t, r)$ equals $\epsilon_{\text{set}}(t, k, r)$.

\subsection{Circuit complexity} \label{sec:circuit_complexity}

As highlighted in Section \ref{subsec:quant_algor}, the solution of the wave equation \eqref{eq:1d_wave-q} may be obtained by evolving the Hamiltonian \eqref{eq:hamil_wave_mod}, i.e., by computing  $e^{-i H_k t} \vec{\psi}(t=0)$. In this subsection, we estimate the computational complexity of implementing the propagator $e^{-i H_k t}$, specifically in terms of the required number of one and two qubit gates.

Section \ref{ses:Time_evolution} describes the implementation of $e^{-i H_k t}$ using the Trotter formula \eqref{eq:Trotter_formula-p}. The complexity of this approach is given by $g = r g_1$, where $r$ is the number of Trotter steps, and $g_1$ is the number of one and two qubit gates in one Trotter step.

In following we consider that matrix $B_k$ is of size $2^n \times 2^n$. To estimate $g_1$ recall that we decompose the Hamiltonian as $H_k = \sum_{\gamma=1}^{\Gamma_k} H_{k, \gamma}$, where each $H_{k, \gamma}$ represents a sum of mutually commuting Pauli strings with respective weights, making $\Gamma_k = s(d,n)$, with $s(k,n)$ given by Proposition \ref{theor:number_sets}. 

The mutually commuting Pauli strings in each $H_{k, \gamma}$ can be simultaneously diagonalized using diagonalization operators $R_{k,\gamma}$ as shown in \cite{Kawase, Berg}. The $R_{k,\gamma}$ operators are derived from elements of the Clifford algebra, such as the  single qubit Hadamard gates and two-qubit $CX$ and $CZ$ gates, as discussed in Section \ref{sec:best_praxis}. The diagonalization operator achieves a gate count that scales as $O(n^2)$. Following diagonalization, we approximate one Trotter step (for simplicity, shown here in first-order form) as
\begin{equation}
    e^{-i \frac{t}{r} \sum_{\gamma=1}^{\Gamma_k} H_{k, \gamma}} \approx \prod_{\gamma=1}^{\Gamma_k} R^\dagger_{k,\gamma} e^{-i \frac{t}{r} \Lambda_{k, \gamma}} R_{k,\gamma},
\end{equation}
where $\Lambda_{k, \gamma}$ is a diagonal matrix. Welsh et al. \cite{Welsh} state that the diagonal exponentiation circuit uses $O(N')$ gates, where $N' < 2^n$, without the need for auxiliary qubits.

The Trotter step's gate count comprises $s(k,n)$ diagonal exponentials and $2s(k,n)$ diagonalization operators.

We also assume that $B_k(c)$ exhibits sparsity, meaning $k << 2^n$, a typical case, from which we derive the subsequent scaling
\begin{equation}
    s(k,n) \approx k(1+n-\log_2(k)) \approx O(kn).
\end{equation}

Combining these findings, we determine the gate complexity to be
\begin{equation}
    g = O(r (2n^2 + 2^n) s(k,n)) = O(2^n r k n),
\end{equation}
where $r$ represents the count of Trotter steps, $2^n$ indicates the dimension of $B_k(c)$, and $2k$ denotes the precision level of the initial derivative approximation.

The Trotter step count $r$ is derived from equation \eqref{eq:scaling_tr}, with $H_k$ expressed as a sum of $\Gamma = s(k,n)$ matrices formed from commuting sets. Given  $\norm{H_k} \leq |\frac{2 k c_{\max}}{h}| = O(kN)$ the total gate count $g$ becomes
\begin{equation}
    g = O\left(t 4^n k^3 n^2 5^{\lfloor p/2 \rfloor} \left( \frac{2}{5} \frac{t 2^n k^2 n}{\epsilon_{tr}}\right)^{1/p} \right).
    \label{eq:resulted_gates_k}
\end{equation}
Using \eqref{eq:scaling_ds}, we can substitute $k$ with the discretization error $\epsilon_{ds}$, so that $k = O\left( \frac{\log_2 (1/\epsilon_{ds})}{2n} \right)$. Then,
\begin{equation}
    g = O\left(t 4^n \frac{\log_2^3 (1/\epsilon_{ds})}{n} 5^{\lfloor p/2 \rfloor} \left( \frac{2}{5} \frac{t 2^n \log_2^2 (1/\epsilon_{ds})}{n \epsilon_{tr}}\right)^{1/p} \right).
    \label{eq:resulted_gates_eps_ds}
\end{equation}
The primary influences on complexity are the diagonal matrix exponential's implementation and the Hamiltonian norm. In our simulations, we employ a Trotter order of $p = 2$ and $N = 2^n$, yielding
\begin{equation}
    g = O\left( \frac{t^{1.5} N^{2.5}}{ n^{1.5}} \log_2^4(1/\epsilon_{ds}) \sqrt{1/\epsilon_{tr}} \right).
    \label{eq:resulted_gates_p-2}
\end{equation}
Equation \eqref{eq:resulted_gates_p-2} shows through complexity analysis that Trotter error have a more substantial impact on the total gate count than discretization error. To reduce the total gate count, we recommend maximizing the Trotterization error while minimizing the discretization error to achieve the required overall solution accuracy.

\section{Conclusion} \label{sec:conclusions}

Our study proposes a method for decomposing $d$-band matrices in Pauli basis, with significant benefits for practical applications. By exclusively considering non-zero Pauli string candidates, our approach outperforms direct strategies based on the matrix multiplication and enables concurrent diagonalization of commuting subsets.

The proposed decomposition method for $d$-band matrices in Proposition \ref{theor:sets} also works for full $2^n \times 2^n$ matrices, when we set $d = 2^{n-1}$ in formulae \eqref{eq:unique_solutions}-\eqref{eq:unique_solutions_0}. Proposition \ref{theor:number_sets} and equation \eqref{eq:number_sets} yield the number of sets for the decomposition, providing a closed-form expression for calculating the number of Pauli strings that may play a role in the decomposition.

Using the one-dimensional wave equation as a model, we have carried out a number of numerical experiments. First, we investigated the implementation of Dirichlet boundary conditions in high-order accuracy schemes (Section \ref{subsec:boundary_cond}) and found that the boundary extension methods based on implicit assumption of function being zero outside of the integration interval are insufficient. Instead, anti-symmetric extension at the boundaries yields precise outcomes, which is essential for accurate simulations. This expansion needs to be directly integrated into the Hamiltonian so that it can be represented using weighted Pauli strings, which will contribute in the computation process.

Secondly, we confirmed that utilizing higher-order discretization improves the accuracy of solutions in quantum algorithms. This improvement is due to the reduction of errors in the approximation of derivatives by employing higher-order numerical schemes. This is shown in Figure \ref{fig:standing_comparison_error_qubits} for up to 8 qubits and the improvement is exponential in discretization order $\kappa$. 
Significantly, accuracy may be retained if decreased qubit count is accompanied by associated increase in discretization order.

Nevertheless, enhancing solution precision leads to greater gate complexity owing to a rise in Trotter steps, necessitating a balance between accuracy and computational resources.
Figure \ref{fig:fixed_trotter_error} demonstrates that a limitation in Trotterization precision consequently caps the overall accuracy of the solution, even if the accuracy of discretization allows for a more precise result. Thus, enhancing discretization precision necessitates a proportional rise in the number of Trotter steps to achieve improved solution accuracy.

Our results suggest that the precision of the solution is constrained by both the Trotterization algorithm and the discretization scheme, and these contributions operate almost independently (see Proposition \ref{theor:error_inter}). Per equation \eqref{eq:scaling_tr}, the Trotterization error shows a linear relation with the Hamiltonian norm, which varies slightly depending on the discretization scheme. Thus, the discretization scheme remains largely irrelevant to the Trotterization's precision, indicating that these two factors of overall accuracy function separately.

The findings reveal the necessity of accounting for both the numerical scheme order and Trotter error, which leads to a rise in gate complexity when executing quantum algorithms. 

However, by disentangling these components, it is possible to enhance the accuracy of solutions within specified error bounds and gate complexity constraints.

\section{Acknowledgements}
\textit{Competing interests.} The authors declare no competing interests. 
\textit{Author contributions.} All authors conceived and developed the theory and design of this study and verified the methods. All research is original work, authors acknowledge the use of generative AI for grammar and paraphrasing. 

The authors acknowledge the use of Skoltech’s Zhores supercomputer \cite{zhores} for obtaining the numerical results presented in this paper.

The authors express gratitude to Daniil Rabinovich, Ernesto Campos, Konstantin Antipin and Stepan Manukhov for suggestions and fruitful discussions.

The PYTHON code for the numerical experiment presented in all Figures is available on \href{https://github.com/barseniev/General-Matrix-Decomposition}{GitHub} \cite{github_page}.

\bibliography{bibliography}

\pagebreak

\appendix

\section{Proofs} \label{sec:appendix_proof}

\subsection{The Order of Binary Strings}
\label{sec:appendix_bin_order}
As explained earlier in Section \ref{sec:notation}, we adopt the standard MSB (Most Significant Bit) format for numbers represented by bit strings ($\MBinCut{n}{x}$ function). For example, $\MBinCut{8}{18}$ results in the bit string $00010010$.

This appendix utilizes the LSB notation, which transforms a number $b$ into a bit string of length $a$ using the operator $\BinCut{a}{b}$, ensuring the sequence begins with the least significant bit on the left. This creates a bit-reversed sequence, for instance, as compared to the earlier example: $\BinCut{8}{18}=01001000$.

The LSB convention simplifies the task of determining the remainder $r_k$ produced by adding two numbers in LSB format by following the bits from left to right.

\subsection{Proof of Proposition  \ref{theor:sets} for decomposition of general tridiagonal matrix}
\label{app: sec: theorem_general}
Two binary strings $\mathbf{a}=(a_1,a_2,\ldots,a_n,0)$ and $\mathbf{b}=(b_1,b_2,\ldots, b_n,0)$, corresponding to decimal numbers $a$ and $b$ can be added to obtain a binary string $\mathbf{c} = \mathbf{a} + \mathbf{b} = (c_1,c_2,\ldots,c_{n+1}), $ we append the appropriate number of zeroes to make sure that the strings we are adding have the same length. The following formula holds:
\begin{equation}
 c_k =(a_k+b_k+r_k(a,b))(\Mod{2}) = a_k \oplus b_k \oplus r_k(a,b).
\end{equation}
Here, the carry-over term $ r_1(a,b)=0, r_2(a,b)=a_1b_1$ and 

 $r_{k}(a,b)=$
\begin{align*}
     &= \left( \sum^{k-2}_{m=1} a_{m}b_{m}\left( \prod ^{k-1}_{j=m+1}\left( a_{j} + b_{j}\right) \right) + a_{k-1}b_{k-1} \right) \Mod{2} \\
    & = \bigoplus^{k-2}_{m=1} a_{m}b_{m}\left( \prod ^{k-1}_{j=m+1}\left( a_{j} \oplus b_{j}\right) \right) \oplus a_{k-1}b_{k-1}, \quad k>2.
    \label{eq:carry-over-term}
\end{align*}

In paper \cite{tridiagonal} we had the following Proposition in the Appendix:

\begin{restatable}{proposition}{lemmaCondition}
\label{lem: lem 1}
    Let $B$ be an upper $d$-diagonal matrix. If a Pauli string $P$ enters the Pauli string decomposition of matrix $B\in  \mathbb{C}^{2^n \cross 2^n}$ non-trivially then $\exists p \in \{0, \dots, 2^{n}-1-d\}$:
    \begin{equation}\label{eq: band}
        \mathbf{p}+\mathbf{d} = \mathbf{p}\oplus \mathbf{x} ,
    \end{equation}
    where $\mathbf{x}$ and $\mathbf{z}$ are such that $P = \hat{W}(\mathbf{x},\mathbf{z})$. 
\end{restatable}
\begin{remark}
 As described in the Appendix \ref{sec:appendix_bin_order}  when we input  $\mathbf{z}$ and $\mathbf{x}$ into the Walsh operator we reverse the binary strings prior to evaluation.
\end{remark}

Writing \eqref{eq: band} explicitly for $k$-th bit, we obtain:
\begin{equation*}
p_{k}\oplus d_{k}\oplus r_{k}(p,d)= p_{k}\oplus x_{k}\\ \Rightarrow x_{k}=d_k \oplus r_{k}(p,d).
\end{equation*}
Therefore, $x_1 = d_1$, $x_2 = d_2 \oplus p_{1}d_{1}$, and for $k>2$ we obtain
\begin{align}
x_{k}=d_{k}\oplus \bigoplus ^{k-2}_{m=1}p_{m}d_{m}\left( \prod ^{k-1}_{j=m+1}\left( p_{j} \oplus d_{j}\right)\right) \oplus p_{k-1}d_{k-1}.
\label{sols}
\end{align}
From this formula one can derive the recurrent relation for $k > 1$:
\begin{equation}
x_{k+1} \oplus d_{k+1} = (d_{k} \oplus x_{k})(p_{k} \oplus d_{k}) \oplus p_{k} d_{k}.
\label{eq:rec_rel}
\end{equation}

Let $l=\lceil \textit{$log_2$}d \rceil$ and $$P_n^m= \underbrace{(\overbrace{1,1, \ldots 1}^m, 0, \ldots, 0)}_n, \quad m=\overline{0,n}.$$

Let $S_l \in \mathbb{B}^n$ a binary variable string of length l.

Let us prove the following three Lemmas:

\begin{lemma}\label{lemma_all}
    All solutions $\mathbf{x}$ of \eqref{eq: band} have the form
    \begin{equation}
        \mathbf{x} = S_l*P_{n-l}^m, \quad m=\overline{0,n-l}.
    \end{equation} 
    
\end{lemma}
\begin{proof}

If $d\neq2^t$ for any $t$ then $l=\BinL{d}.$ Therefore, $d_{l+1},\ldots,d_n=0.$ Hence, substituting $k=l+1$ in \eqref{eq:rec_rel} we get $x_{l+2}= x_{l+1}p_{l+1}$.

Similarly, substituting $k = \overline{l+2,n}$ in \eqref{eq:rec_rel} we obtain,
    \begin{equation}
    x_k= \left( \prod ^{k-1}_{j=l+1}p_j \right)x_{l+1}, \quad k=\overline{l+2,n}.
    \label{eq:Pyr}
    \end{equation}
Since, $p_j$ can be chosen arbitrarily we get $x= S_l*P_{n-l}^m$ where $m$ since, if some $x$ is zero all the following $x$'s are zeros as well.

If $d=2^t$ for some $t$ from \eqref{eq:rec_rel} it follows that $x_j=0$ for $j=\overline{0,l}$, $x_{l+1}=1.$ Further from \eqref{eq:Pyr} it follows that $(x_{l+2},x_{l+3},\ldots,x_{n})= P_{n-l-1}^m$ for certain $m$.

Thus, in either case the solution has the form $x= S_l*P_{n-l}^m$.
\end{proof}

 For a fixed $l$ we call a solution $\mathbf{x}$ of \eqref{eq: band} \textbf{novel} if it is not a solution of \eqref{eq: band} for $l'<l.$

\begin{lemma}\label{lemma_new}
    For fixed $l$, all novel solutions $\mathbf{x}$ of \eqref{eq: band} have the form
    \begin{equation}
        \mathbf{x} = S_l*P_{n-l}^m,\quad m=\overline{1,n-l}.
    \end{equation}
\end{lemma}
\begin{proof}
From Lemma $\ref{lemma_all}$ it follows that all solutions of \eqref{eq: band} for a given $l$ are either $\mathbf{x}=S_l*P_{n-l}^0$ or  $\mathbf{x} = S_l*P_{n-l}^m,\quad m=\overline{1,n-l}.$

We will prove that all the solutions have the form $\mathbf{x} = S_l*P_{n-l}^m,\quad m=\overline{1,n-l}.$

It holds for $l=0$ as shown in \cite{tridiagonal}.

Assume it to be true for a certain $l-1$, let us prove that it is also true for $l$.

Assume the contrary i.e. let for a fixed $l$ there exist a solution $\mathbf{x}$ that satisfies equation \eqref{eq: band} and has the form $\mathbf{x}=S_l*P^0_{n-l}$. Let us check the $l$-th coordinate of $\mathbf{x}$. If it is equal to 1, the solution can be rewritten in the form $S_{l-1}*P^1_{n-l+1}$ which has appeared before according to our assumption. Otherwise, consider the $l-1$-th coordinate of $\mathbf{x}$ and use the same argument. If 1 never appears in the string we get the solution $\mathbf{x}=\vec{0}$ which is considered separately (diagonal case).

Thus, all the solutions have the form $\mathbf{x}=S_l*P_{n-l}^m, m=\overline{1,n-l}$.
\end{proof}

\begin{lemma}\label{lemma_new_0}
    For fixed $l>0$, all novel solutions $\mathbf{x}$ of \eqref{eq: band}  have the form
    \begin{equation}
        \mathbf{x}= S_{l-1}*0*P_{n-l}^m, \quad m=\overline{1,n-l}.
    \end{equation}
\end{lemma}

\begin{proof}
    Similar to the above Lemma \ref{lemma_new} assume that the assumption is true for $l-1$. For $l$ pick a solution $\mathbf{x}$ that satisfies equation \eqref{eq: band} and has the form $\mathbf{x}=S_{l-1}*1*P_{n-l}^m.$ Now check the $(l-1)$-th coordinate of $\mathbf{x}$. If it is equal to 0, the solution can be rewritten in the form $S_{l-2}*0*P^m_{n-l+1}$ which has appeared before according to our assumption. Otherwise, consider the $(l-1)$-th coordinate of $\mathbf{x}$ and use the same argument. If $0$ never appears in the string we get a solution which corresponds to the case $d=1 ~ (l=0)$ and hence has appeared before and is not new.

    Thus, all novel solutions have the form 
    $\mathbf{x}~=~S_{l-1}~*~0~*~P_{n-l-1}^m$.
\end{proof}

We will now prove the subsequent Proposition, leading to Proposition \ref{theor:sets}.

\begin{proposition}
\label{thm:main}
For given $d$ and $n$, each novel solution $\mathbf{x}$ of equation \eqref{eq: band} is encoded by a bit string 
    \begin{equation}\label{eq: new_solutions}
        \mathbf{x} = \Bin{D-d}*P^m_{n-l}, \quad m=\overline{0,n-l},
    \end{equation}
   
    where $D=2^\BinL{d}=2^l$.
\end{proposition}
\begin{proof}
 First, we calculate the total possible number of novel solutions. Second, we construct strings that satisfy the equation and show that the number of them is the same as the number of possible strings, therefore we provide all solutions. 

From Lemma \ref{lemma_new_0}, it is established that for a given $l$, all possible novel solutions take the form 
\begin{equation*}
    \mathbf{x}= S_{l-1}*0*P_{n-l}^m, \quad m=\overline{1,n-l},
\end{equation*}
where $S_{l-1} \in \mathbb{B}^{l-1}$ and there are $2^{l-1}$  possible configurations of the string $S_{l-1}$.

It can be seen that for a fixed value $l > 0$ (for a fixed binary length) there are $2^{l-1}$ possible strings $d$.

That number matches the number of possible strings $S_{l-1}$. We show that each $d$ produces novel solutions (up to $d = 2^{n-1} $) and claim that each $d$ uniquely corresponds to a distinct string $S_{l-1}$. Consequently, for a given $d$, all novel solutions are described by the following $n-l$ binary strings:
\begin{equation}
    \mathbf{x} = S_{l-1}(d)*0*P_{n-l}^m, \quad m=\overline{1,n-l},
    \label{eq:sol_fixed_d}
\end{equation}
where $S_{l-1}(d)$ is fixed for each specified $d$.

If $S_{l-1}(d)$ were arbitrary then the total number of possible strings would be $2^{l-1}*1*(n-l)=2^{l-1}*(n-l).$

Let us consider the strings 
\begin{equation*}
    \mathbf{x} = \BinCut{l-1}{D-d}*0*P^m_{n-l}, \quad m=\overline{0,n-l}.
\end{equation*}
We show that $S_{l-1}$ from \eqref{eq:sol_fixed_d} is $ \BinCut{l-1}{D-d}$.
For a given $l$ and $D=2^l$ the value $D-d$ is unique. Thus, for a given $d$ we can generate   $n-l$ unique solutions. Thus, for each of the $2^l-1$ $d'$s we obtain all the $2^{l-1}*(n-l)$ possible solutions.

To complete the proof, it remains to demonstrate that, for a given $d$, the strings are indeed solutions to the equation $\Bin{p+d} = \Bin{p} \oplus \mathbf{x}$. Once we establish that these strings are solutions, it follows from equation \eqref{eq:sol_fixed_d} that these constitute all the novel solutions for the given $d$.

There are two possible cases, when $d$ is some power of two or not.
The case $d=2^t$ has been considered in Lemma \ref{lemma_all}.
Consider the case when $d \neq 2^t$ for all $ t \in \mathbb{Z}$, thus $l=\BinL{d}$.
From \eqref{eq: band} we have $$\mathbf{x}=(\mathbf{p}+\mathbf{d}) \oplus \mathbf{p}.$$
Let's chose $p=D-d$ where $D=2^\BinL{d}=2^l$.

We use the following equality for binary addition rule:
\begin{equation*}
    (\mathbf{a}\oplus\mathbf{b})_k = \left\lfloor\frac{a}{2^{k-1}}\right\rfloor\! \bmod{2} \oplus \left\lfloor\frac{b}{2^{k-1}}\right\rfloor\! \bmod{2}, \quad \forall a, b \in \mathbb{B}.
\end{equation*}

Thus, the $k$-th coordinate of $\mathbf{x}$ is equal to 
\begin{equation*}
    x_k = \left( \left\lfloor \frac{D}{2^{k-1}} \right\rfloor \bmod 2 \right)  \oplus \left( \left\lfloor\frac{D-d}{2^{k-1}} \right\rfloor\bmod 2 \right).
\end{equation*}
As $D$ is a some power of two, then
$$
\left\lfloor \frac{D}{2^{k-1}} \right\rfloor \bmod 2 = 0
$$
and  
$$
x_k =  \left\lfloor\frac{-d}{2^{k-1}} \right\rfloor\bmod 2.
$$
Using the properties of the floor and ceiling function 
\begin{equation*}
     \; \lfloor -a \rfloor \bmod{2} =  \lceil a \rceil \bmod{2}, \quad \forall a \in \mathbb{B}
\end{equation*}

we obtain 
\begin{equation}
    x_k = \left\lceil\frac{d}{2^{k-1}}\right\rceil \bmod 2, 
\end{equation}
which is $0$ for $l-1<k$:
\begin{equation}
    x_k=\left\{\begin{array}{cl}\left\lceil\frac{d}{2^{k-1}}\right\rceil \bmod 2, & \text { if } k \leqslant l+1, \\ 0, & \text { if } k>l+1.\end{array}\right.
\end{equation}

It also can be seen that since $\BinL{d}=l$,
 $x_l=0$ and $x_{l+1}=1$.
In a similar way, by substituting $p=2^jD-d$ we obtain rest of the solutions of the type $\Bin{D-d}*P$. 
  
\end{proof}

Writing out $\Bin{D-d}$ and $P$ explicitly, in terms of $\MBin{D-d}$ we have:

\decompositionSets*

\subsection{Proof of Proposition  \ref{theor:number_sets} for decomposition of the general tridiagonal matrix}

\decompositionNumSets*
\begin{proof}
According to Proposition \ref{theor:sets}, the total number of sets for the $d$ band matrix $B$ of size $2^n \times 2^n$, including the diagonal, is given by the formula:
\begin{equation}
    s = 1 + \sum_{k=1}^d (n - \lceil \log_2 k \rceil)
\end{equation}

To compute this sum explicitly, we first note that the term $\lceil \log_2 k \rceil$ repeats exactly $2^{\lceil \log_2 k \rceil-1}$ times. Therefore, we can express the sum as:
\begin{equation}
    \sum_{k=1}^d (\lceil \log_2 k \rceil) = \sum_{r=1}^{\lfloor \log_2 d \rfloor} r 2^{r-1} + \lceil \log_2 d \rceil (d- 2^{\lfloor \log_2 d \rfloor}),
\end{equation}

where the first term accounts for sum from $k = 1$ to $2^{\lfloor \log_2 d \rfloor}$, and the second term accounts for the remaining terms, i.e. for sum from $k = 2^{\lfloor \log_2 d \rfloor} + 1$ to $d$. Next, we compute the first sum explicitly:
\begin{equation}
    \sum_{r=1}^{\lfloor \log_2 d \rfloor} r 2^{r-1} = 2^{\lfloor \log_2 d \rfloor}( \lfloor \log_2 d \rfloor - 1) + 1
\end{equation}

Combining this together, we get 
\begin{equation}
    s = 2^{\lfloor \log_2 d \rfloor}( \lceil \log_2 d \rceil + 1 - \lfloor \log_2 d \rfloor) + d(n - \lceil \log_2 d \rceil).
\end{equation}

It is straightforward to verify that in both cases, whether $d$ is a power of $2$ or not, we obtain the formula presented in the Proposition, i.e.
\begin{align}
\begin{split}   
    s &= 2^{\lfloor \log_2 d \rfloor + 1} + d(n - \lceil \log_2 d \rceil - 1) \\
      &= 2^{\BinL{d}} + d(n - \BinL{d}).
\end{split}
\end{align}

\end{proof}

\subsection{Proof of Proposition  \ref{theor:error_inter} for error impact}

\errorInter*
\begin{proof}

Let us define a new operator 
\begin{equation}
    H_{d/dx} = 
    \begin{pmatrix}
        0 & \frac{d}{dx}(c(x) \cdot) \\
        -c(x)\frac{d}{dx} & 0
    \end{pmatrix},
\end{equation}
where operator $\frac{d}{dx}(c(x) \cdot)$ acts on function $f(x)$ as $\frac{d}{dx}(c(x) \cdot)f(x)= \frac{d}{dx}(c(x) f(x))$.
The action of $H_{d/dx}$ on a vector function may be written as
\begin{equation*}
\begin{split}
    H_{d/dx} 
    \begin{pmatrix}
        f_1(x) \\
        f_2(x)
    \end{pmatrix}
    & =
    \begin{pmatrix}
        \frac{d}{dx}(c(x) f_2(x)) \\
        -c(x)f'_1(x)
    \end{pmatrix} \\
    & =
    \begin{pmatrix}
        c'(x) f_2(x) + c(x) f'_2(x)  \\
        -c(x)f'_1(x)
    \end{pmatrix}.
\end{split}
\end{equation*}

The same operator of bigger size can be defined as 

\begin{equation*}
    H_{d/dx} = 
    \begin{pmatrix}
        0 & 0 & \frac{d}{dx}(c(x) \cdot) & 0\\
        0 & 0 & 0 & \frac{d}{dx}(c(x) \cdot)\\
        -c(x)\frac{d}{dx} & 0 & 0 & 0\\
        0 & -c(x)\frac{d}{dx} & 0 & 0
    \end{pmatrix},
\end{equation*}
which acts on a vector function, as
\begin{equation*}
    H_{d/dx} 
    \begin{pmatrix}
        \vec{f}_1 \\
        \vec{f}_2
    \end{pmatrix}
    = 
    H_{d/dx} 
    \begin{pmatrix}
        f_1(x_1) \\
        f_1(x_2) \\
        f_2(x_1) \\
        f_2(x_2)
    \end{pmatrix}
    =
    \begin{pmatrix}
        \frac{d}{dx}(c(x_1) f_2(x_1)) \\
        \frac{d}{dx}(c(x_2) f_2(x_2)) \\
        -c(x_1)f'_1(x_1) \\
        -c(x_2)f'_1(x_2)
    \end{pmatrix}.
\end{equation*}

Following \cite{Costa} we can obtain the exact solution $u(t,x)$ of the wave equation \eqref{eq:1d_wave-q} as the first component of
\begin{equation}
    \begin{pmatrix}
        u(t,x) \\
        \Tilde{u}(t,x)
    \end{pmatrix}
    = e^{-i t H_{d/dx} } 
    \begin{pmatrix}
        u_0(x) \\
        0
    \end{pmatrix},
\end{equation}
where $\Tilde{u}(t,x)$ is a component we are not interested in.
We recall that the numerical solution produced by proposed algorithm can be written as 
\begin{equation}
    \begin{pmatrix}
        \vec{\phi}_V(t) \\
        \vec{\phi}_E(t)
    \end{pmatrix}
    = S_{p}^r(t/r) 
    \begin{pmatrix}
        \vec{u}_0 \\
        \vec{0}
    \end{pmatrix},
\end{equation}
where $S_{p}^r(t/r)$ is the Trotterized propagator $e^{-i t H_{k}}$ given by equation \eqref{eq:Trotter_formula-p} with $H_k$ given by \eqref{eq:hamil_wave}, and $\vec{u}_0$ being discretized initial condition, from \eqref{eq:1d_wave-q}.

Taking into account all mentioned above we can write
\begin{equation*}
\begin{split}
    \epsilon_{\text{ns}}(t) &= \norm{\vec{u}_{\text{sw}}(t)-\vec{\phi}_V(t)} \\
    & \leq \norm{ e^{-i t H_{d/dx}} 
    \begin{pmatrix}
       \vec{u}_0 \\
       \vec{0}
    \end{pmatrix} 
    -  S_{p}^r(t/r) 
    \begin{pmatrix}
       \vec{u}_0 \\
       \vec{0}
    \end{pmatrix}
    }.
\end{split}
\end{equation*}

In what follows, consider that the initial state is normalized, that is, $\norm {\vec{u}_0} = 1$, using this, one can get
\begin{equation*}
    \begin{split}
    \epsilon_{\text{ns}}(t) & \leq \norm{ \left( e^{-i t H_{d/dx} } - e^{-i t H_{k}} + e^{-i t H_{k}} - S_{p}^r(t/r) \right)} \\
        & \leq  \norm{ \left( e^{-i t H_{d/dx}} - e^{-i t H_{k}} \right)}  + \norm{ e^{-i t H_{k}} - S_{p}^r(t/r) }  \\
    \end{split}
\end{equation*}
where the matrix $H_{k}$ is the approximation of operator $H_{d/dx}$, given by \eqref{eq:hamil_wave}. We note that $\frac{d}{dx}(c(x)\cdot)$ can be written in matrix form as $B_k I_c$, while $-c(x)\frac{d}{dx}$ can be written in matrix form as $-I_c B_k = (B_k I_c)^T$, where $I_c$ is diagonal matrix which contains values of $c(x)$ on the diagonal (same as in equation \eqref{eq:laplacian_wave}).

Note that the second term is exactly the Trotterization error $\epsilon_{\text{tr}}(r)$ defined in \eqref{eq:errors_trotter}. 

Before considering the first term we introduce $R~=~H_{d/dx}-H_{k}$ and taking into consideration that the derivative is a linear function, it can easily be checked that the operators $H_{d/dx}$ and $H_{k}$ commute, therefore $R H_{k} = H_{k} R$.

Thus, it follows that
\begin{equation*}
    \begin{split}
        \norm{ e^{-i t H_{d/dx}} - e^{-i t H_{k}}} & = \norm{ e^{-i t (R + H_{k})} - e^{-i t H_{k}} } \\
        & = \norm{  e^{-i t R} e^{-i t H_{k}}  - e^{-i t H_{k}} } \\
        & \leq \norm{  e^{-i t R} - I } \norm{e^{-i t H_{k}}} \\
        & \leq \sum_{j=1}  \frac{\abs{t}^j}{j!} \norm{ R^j } \\
        & \leq \sum_{j=1}  \frac{\abs{t}^j}{j!} \norm{ R }^j = e^{t \norm{R}} - 1,
    \end{split}
\end{equation*}
where we considered $t > 0$. 
The only part left to show is $\norm{R} \leq \epsilon_{\text{ds}}(k) \sqrt{ c_{\max}^2 +1}$.

Let us recall that
\begin{equation*}
    \epsilon_{\text{ds}}(k) = \norm{ \frac{d}{dx} \vec{f} - B_k \vec{f}} = \sqrt{ \sum_{j=1}^N   \abs{ f'(x_{j}) + \hat{B}_k f(x_{j}) }^2} 
\end{equation*}
We note now that action of $R$ on some state of size $2N$ may be written as, 
\begin{equation*}
    R
    \begin{pmatrix}
        f_1(x_1) \\
        \vdots \\
        f_1(x_N) \\
        f_2(x_{1}) \\
        \vdots \\
        f_2(x_{N})\\
    \end{pmatrix} =
    \begin{pmatrix}
        \frac{d}{dx}(c(x_1) f_2(x_1)) - \hat{B}_k (c(x_1) f_2(x_1)) \\
        \vdots \\
        \frac{d}{dx}(c(x_N) f_2(x_N)) - \hat{B}_k (c(x_N) f_2(x_N))  \\
        -c(x_{1})f'_1(x_{1}) + c(x_{1}) B_k f_1(x_{1}) \\
        \vdots \\
        -c(x_{N})f'_1(x_{N}) + c(x_{N}) B_k f_1(x_{N})
    \end{pmatrix}
\end{equation*}

Therefore 
\begin{equation}
    \norm{R \vec{f}} = \sqrt{F_2 + F_1},
\end{equation}
where given formula \eqref{eq:approximation-order} we have
\begin{equation}
    \begin{split}
        F_1 &\equiv \sum_{j=1}^N   \abs{ -c(x_{j})f'_1(x_{j}) + c(x_{j}) \hat{B}_k f_1(x_{j}) }^2 \\
        & \leq  c_{\max}^2 \sum_{j=1}^N   \abs{f'_1(x_{j}) - \hat{B}_k f_1(x_{j}) }^2  = c_{\max}^2 \epsilon_{\text{ds}}^2(k),
    \end{split}
\end{equation}
with $c_{\max} = \max_x \abs{c(x)} = \max_x c(x)$, since $c(x) > 0$, and 
\begin{equation}
    \begin{split}
        F_2 &\equiv \sum_{j=1}^N \abs{  \frac{d}{dx}(c(x_j) f_2(x_j)) - \hat{B}_k (c(x_j) f_2(x_j)) }^2 \\
        & = \sum_{j=1}^N \abs{  \frac{d}{dx}g(x_j) - \hat{B}_k g(x_j) }^2 = \epsilon_{\text{ds}}^2(k),
    \end{split}
\end{equation}
where $g(x) \equiv c(x) f_2(x)$. 

Therefore, we have the following result
\begin{equation}
    \norm{R \vec{f}} \leq \sqrt{ c_{\max}^2 \epsilon_{\text{ds}}^2(k) + \epsilon_{\text{ds}}^2(k)} = \epsilon_{\text{ds}}(k) \sqrt{ c_{\max}^2 +1},
\end{equation}
and using it we have the operator norm as 
\begin{equation}
    \norm{R} = \sup_{\norm{\vec{f}} = 1} \norm{R \vec{f}} \leq \epsilon_{\text{ds}}(k) \sqrt{ c_{\max}^2 +1}.
\end{equation}
Finally, we have the following estimation
\begin{equation*}
    \norm{ e^{-i t H_{d/dx}} - e^{-i t H_{k}}} \leq e^{t \norm{R}} - 1 \leq e^{t \epsilon_{\text{ds}}(k) \sqrt{ c_{\max}^2 +1} } - 1,
\end{equation*}
and the the final result
\begin{equation}
\epsilon_{\text{ns}} (t) \leq \epsilon_{\text{tr}}(t, r) + e^{t \epsilon_{\text{ds}}(k) \sqrt{ c_{\max}^2 +1} } - 1 \approx \epsilon_{\text{tr}}(t, r) + t c_{\max} \epsilon_{\text{ds}}(k),
\end{equation}
where we used 
\begin{equation*}
    e^{t \epsilon_{\text{ds}}(k) \sqrt{ c_{\max}^2 +1} } - 1 = O\left(t c_{\max} \epsilon_{\text{ds}}(k) \right).
\end{equation*}

\end{proof}
\section{Dirichlet boundary conditions for higher order schemes} \label{sec:appendix_boundary}

First, we examine $\hat{B}_1$: with this operator, the right-hand side of equation \eqref{eq:1d_wave-q} at $x = 0$ transforms the wave equation into

\begin{equation*}
    \begin{split}
    \frac{d^2 u}{dt^2}(0) &= \frac{1}{4h^2} \left( c^2(-h)u(-2h) - [c^2(-h) + c^2(h)]u(0) \right.\\
    & \left. + c^2(h)u(2h) \right).
    \end{split}
\end{equation*}

By applying conditions $\frac{d^2}{dt^2} u(t, 0) = 0$ and $u(0) = 0$, we obtain the following relation
\begin{equation}
    c^2(-h)u(-2h) + c^2(h)u(2h) = 0.
    \label{eq:condition_symm}
\end{equation}
For the particular case of a constant wave speed $c$, this condition simplifies to $u(-2h) = -u(2h)$. Substituting $u(-2h) = -u(2h)$ back into equation \eqref{eq:condition_symm} results in $c^2(-h) = c^2(h)$.

We now discus modification of matrix $B_k$ near boundaries, using the condition $u(t, -x) = -u(t, x)$. Note, that since $c(-x) = c(x)$ the speed profile may be incorporated afterwards, the same way as it is done in \eqref{eq:laplacian_wave}. We investigate only the upper left corner of this matrix. The bottom right corner mirrors the same properties except for a reflection about the anti-diagonal.

To illustrate the approach it is better to show it on a small example, the generalization is straightforward. Therefore, in what follows, we consider left boundary conditions only, in the example there the derivative approximation is characterized by $k=3$, that is the scheme may be written as $(-b_{3},-b_{2},-b_{1},0,b_{1},b_{2},b_{3})$, see Table \eqref{tab:Approx-M1} for exact coefficients.

Below we show the extension of such operator outside of boundaries depicted by the dashed line.

\begin{figure}[!ht]
\centering
\includegraphics[width=\linewidth]{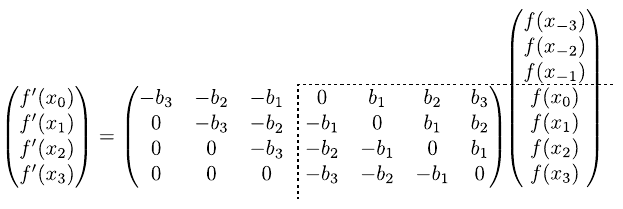}
\end{figure}

In general, we do not know anything about the points $f(x_{-j})$, therefore we will use the assumption derived previously $f(-x) = f(x)$, that is antisymmetry of the function. By employing it, we can write matrix $B_k$ in the following form, where changes to the original matrix are shown with red

\begin{equation*}
    \begin{pmatrix}
    f'(x_{0}) \\
    f'(x_{1}) \\
    f'(x_{2}) \\
    f'(x_{3})
    \end{pmatrix}
=
    \begin{pmatrix}
    0 & b_{1} \textcolor{red}{+ b_{1}} & b_{2} \textcolor{red}{+ b_{2}} & b_{3} \textcolor{red}{+ b_{3}}\\
    -b_{1} & 0 \textcolor{red}{+ b_{2}} & b_{1} \textcolor{red}{+ b_{3}} & b_{2} \\
    -b_{2} & -b_{1} \textcolor{red}{+ b_{3}} & 0 & b_{1} \\
    -b_{3} & -b_{2} & -b_{1} & 0 \\
    \end{pmatrix}
    \begin{pmatrix}
    f(x_{0}) \\
    f(x_{1}) \\
    f(x_{2}) \\
    f(x_{3})
    \end{pmatrix}.
\end{equation*}

Same argument can be made for the second derivative, with one adjustment. Since, the function $f(x)$ is antisymmetric, the derivative of this function should be symmetric, that is $f'(x) = f'(-x)$. We show the result below

\begin{equation*}
    \begin{pmatrix}
    f''(x_{0}) \\
    f''(x_{1}) \\
    f''(x_{2}) \\
    f''(x_{3})
    \end{pmatrix}
=
    \begin{pmatrix}
    0 & b_{1} \textcolor{red}{- b_{1}} & b_{2} \textcolor{red}{- b_{2}} & b_{3} \textcolor{red}{- b_{3}}\\
    -b_{1} & 0 \textcolor{red}{- b_{2}} & b_{1} \textcolor{red}{- b_{3}} & b_{2} \\
    -b_{2} & -b_{1} \textcolor{red}{- b_{3}} & 0 & b_{1} \\
    -b_{3} & -b_{2} & -b_{1} & 0 \\
    \end{pmatrix}
    \begin{pmatrix}
    f'(x_{0}) \\
    f'(x_{1}) \\
    f'(x_{2}) \\
    f'(x_{3})
    \end{pmatrix}.
\end{equation*}

It can be seen that the only difference between these two cases is the sign of modification. It is also worth noting, that we added boundary conditions but lost antisymmetric property if matrix approximation of derivative operator. Combining this together we can write the right hand side of the wave equation \eqref{eq:1d_wave-q} as

\begin{multline*}
    \begin{pmatrix}
    0 & b_{1} \textcolor{red}{- b_{1}} & b_{2} \textcolor{red}{- b_{2}} & b_{3} \textcolor{red}{- b_{3}}\\
    -b_{1} & 0 \textcolor{red}{- b_{2}} & b_{1} \textcolor{red}{- b_{3}} & b_{2} \\
    -b_{2} & -b_{1} \textcolor{red}{- b_{3}} & 0 & b_{1} \\
    -b_{3} & -b_{2} & -b_{1} & 0 \\
    \end{pmatrix} I_q \\ \times 
    \begin{pmatrix}
    0 & b_{1} \textcolor{red}{+ b_{1}} & b_{2} \textcolor{red}{+ b_{2}} & b_{3} \textcolor{red}{+ b_{3}}\\
    -b_{1} & 0 \textcolor{red}{+ b_{2}} & b_{1} \textcolor{red}{+ b_{3}} & b_{2} \\
    -b_{2} & -b_{1} \textcolor{red}{+ b_{3}} & 0 & b_{1} \\
    -b_{3} & -b_{2} & -b_{1} & 0 \\
    \end{pmatrix}
    \begin{pmatrix}
    f(x_{0}) \\
    f(x_{1}) \\
    f(x_{2}) \\
    f(x_{3})
    \end{pmatrix}.
\end{multline*}

As was mentioned before, the key idea we employ in this work is the representation of the right hand side of the wave equation in the form $-BB^T$, with some matrix $B$. To realize it, we note that since $f(x_0) = 0$ we can change the first column of the first derivative approximation arbitrary. We also note that $f'(x_0) \neq 0$, so we can not use the same argument for the second derivative. Having this we change the first column of the first derivative approximation with zeros. Moreover we introduce coefficients $\alpha=1$, $\beta=2$ and rewrite the equation above as

\begin{multline*}
    \begin{pmatrix}
    0 & \textcolor{red}{0} & \textcolor{red}{0} & \textcolor{red}{0}\\
    -\textcolor{red}{\alpha}b_{1} & 0 \textcolor{red}{- b_{2}} & b_{1} \textcolor{red}{- b_{3}} & b_{2} \\
    -\textcolor{red}{\alpha}b_{2} & -b_{1} \textcolor{red}{- b_{3}} & 0 & b_{1} \\
    -\textcolor{red}{\alpha}b_{3} & -b_{2} & -b_{1} & 0 \\
    \end{pmatrix} I_q \\ \times 
    \begin{pmatrix}
    0 & \textcolor{red}{\beta}b_{1} & \textcolor{red}{\beta}b_{2} & \textcolor{red}{\beta}b_{3} \\
    \textcolor{red}{0} & 0 \textcolor{red}{+ b_{2}} & b_{1} \textcolor{red}{+ b_{3}} & b_{2} \\
    \textcolor{red}{0} & -b_{1} \textcolor{red}{+ b_{3}} & 0 & b_{1} \\
    \textcolor{red}{0} & -b_{2} & -b_{1} & 0 \\
    \end{pmatrix}
    \begin{pmatrix}
    f(x_{0}) \\
    f(x_{1}) \\
    f(x_{2}) \\
    f(x_{3})
    \end{pmatrix}.
\end{multline*}

Matrix $I_q$ is diagonal and the coefficients $\alpha$ and $\beta$ are present in the resulted vector only as product, that is $\alpha \beta = 2$. Therefore, we can make them equal $\alpha = \beta = \sqrt{2}$. This way we indeed can see that the right hand side of wave equation with incorporated boundary conditions can be written as $-B_k(c)B_k(c)^T$.
\newpage
\onecolumngrid

\section{program listings} \label{sec:appendix_listings}

\subsection{Listing to compute formulae in proposition  \ref{theor:sets} and \ref{theor:number_sets} for the decomposition of a general tridiagonal matrix}
\label{list:unique_solutions}

\begin{lstlisting}
#include <limits.h> // For CHAR_BIT and __builtin_clz

/**
 * @brief Generates subsets based on given parameters and fills the output array.
 * 
 * This function calculates subsets of a given number of qubits with a specified
 * matrix bandwidth. It optionally handles a special BBbar case. The results are
 * stored in the provided output array.
 * 
 * @param numQubits       The number of qubits.
 * @param matrixBandwidth The matrix bandwidth.
 * @param isBBbarCase     Boolean flag indicating if the BBbar case should be handled.
 * @param x               Output array to store the generated subsets.
 * @return long           The number of subsets generated.
 */
long generate_x(int numQubits,          // in: number of qubits
                int matrixBandwidth,    // in: matrix bandwidth
                int isBBbarCase,        // in: 1 if BBbar case, 0 if general matrix
                unsigned long x[]       // out: fill up this structure
                )                       // returns number of subsets
{
    // Validate input parameters
    if (numQubits <= 0 || matrixBandwidth <= 0) {
        return 0; // Return 0 if the input parameters are invalid
    }

    // Determine the bbit value based on the BBbar case
    unsigned long bbit = isBBbarCase ? (1UL << numQubits) : 0;
    long count = 0; // Initialize count of subsets

    // Loop through each value of k from 1 to the matrix bandwidth
    for (int k = 1; k <= matrixBandwidth; k++) {
        int s = 0;
        // Calculate s as the smallest integer such that 2^s >= k
        while ((1 << s) < k)   s++;    // s = log_2(k) rounded up

        // Loop through each value of j from 1 to (numQubits - s)
        for (int j = 1; j <= numQubits - s; j++) {
            // Calculate the subset and store it in the output array
            x[count++] = bbit | ((1UL << j) - 1) << s | ((1UL << s) - k);
        }
    }

    // Calculate and return the expected number of subsets according to the formula
    int binLd = (sizeof(matrixBandwidth) * CHAR_BIT) - __builtin_clz(matrixBandwidth);
    return (1UL << binLd) + (numQubits - binLd) * matrixBandwidth;
}

\end{lstlisting}

\subsection{Listing to compute Walsh function, Equation \eqref{eq:walsh_oper} from pair $(\mathbf{x}_{k,j},\mathbf{z})$}
\label{list:walsh_oper} 

\begin{lstlisting}
#include <stdio.h>
#include <stdint.h>

// Constants representing Pauli matrices 
// and their corresponding imaginary sign and value tables
const char PAULI_NAMES  [4] = {'I', 'Z', 'X', 'Y'};
const char IMAG_SIGN_TAB[4] = {'+', '+', '-', '-'};
const char IMAG_VAL_TAB [4] = {'1', 'i', '1', 'i'};

/**
 * @brief Calculates the phase (i)^(x.z) and generates a Pauli string index.
 * 
 * This function computes the phase based on the bitwise AND of `x` and `z`, 
 * and fills the `indx` array with indices representing Pauli matrices.
 * 
 * @param n     The length of the Pauli string.
 * @param x     The set label.
 * @param z     The set member.
 * @param indx  The output array to store the Pauli string indices.
 * @return int  The phase (i)^(x.z)
 * first bit in (x,z) is the LSB (least significant bit)
 */
int W(int n,                      // number of qubits 
      unsigned long x,            // set label
      unsigned long z,            // set member
      char indx[]                 // Pauli string length n
     )                    
{
    // Calculate the imaginary component index based on the popcount of x & z
    int imag = __builtin_popcountll(x & z) & 0x3;   // index into the IMAG_SIGN_TAB 
                                                    // and IMAG_VAL_TAB

    // Build the Pauli string index
    indx[0] = PAULI_NAMES[((x << 1) & 0x2) | (z & 0x1)];  // LSB gives indx[0]
    indx[1] = PAULI_NAMES[(x & 0x2) | ((z >> 1) & 0x1)];  // Next bit gives indx[1]

    // Loop through the remaining bits to build the Pauli string index
    for (int i = 2; i < n; i++) {
        indx[i] = PAULI_NAMES[((x >> (i - 1)) & 0x2) | ((z >> i) & 0x1)];
    }

    // Return the phase (i)^(x.z)
    return imag;
}

// Example usage
int main() {
    int n = 4;
    unsigned long x = 0b1010;
    unsigned long z = 0b1100;
    char indx[4];

    int phase = W(n, x, z, indx);

    printf("Phase: %d\n", phase);
    printf("Pauli String: ");
    for (int i = 0; i < n; i++) {
        printf("%c", indx[i]);
    }
    printf("\n");

    return 0;
}
\end{lstlisting}

\subsection{Listing to compute decomposition weights, Equation \eqref{eq:calc_weight} from pair $(\mathbf{x}_{k,j},\mathbf{z})$}
\label{list:calc_weight} 

\begin{lstlisting}
#include <cmath>
#include <cstddef>

// Function prototypes for B(), assuming it is defined elsewhere
double B(long k, long m);

// Constants for bitwise operations
const int BIT_REFLECTION = 0x1;
const int BIT_NEGATIVE_I = 0x2;
const int BIT_MASK = 0x3;

/**
 * Function to calculate the decomposition weight for a real symmetrical matrix.
 *
 * Parameters:
 * - num_qubits (int): The number of qubits.
 * - num_diagonals (long): The number of diagonals to iterate over.
 * - x_label (long): The set label.
 * - z_member (long): The set member.
 *
 * Returns:
 * - double: The calculated decomposition weight.
 *
 * Algorithm:
 * 1. Calculate the matrix size as 2^num_qubits.
 * 2. Determine the imaginary component based on the bitwise AND of x_label and
 *    z_member.
 * 3. Calculate the signs for reflection around the main diagonal and the
 *    presence of -1 or -i.
 * 4. Iterate over all diagonals up to num_diagonals and elements of each
 *    diagonal.
 * 5. Calculate the local set label and skip if it does not match the set label.
 * 6. Determine the sign for the set member and calculate the total sign.
 * 7. Retrieve the values from the band matrices and update the weight.
 * 8. Normalize the weight by the matrix size and return it.
 */
double w_term(int num_qubits, long num_diagonals, long x_label, long z_member) 
{
    size_t matrix_size = 1LL << num_qubits; // Calculate the matrix size as 2^num_qubits

    // Calculate the imaginary component based on bitwise AND of x_label and z_member
    long imag_component = __builtin_popcount(x_label & z_member) & BIT_MASK;

    // Determine the sign for reflection around the main diagonal
    long sign_reflection = (imag_component & BIT_REFLECTION) ? -1 : 1;

    // Determine the sign for the presence of -1 or -i
    long sign_negative_i = (imag_component & BIT_NEGATIVE_I) ? -1 : 1;

    double weight = 0.0;

    // Iterate over all diagonals up to num_diagonals
    for (long diagonal = 0; diagonal <= num_diagonals; diagonal++) {
        // Iterate over all elements of the current diagonal
        for (long element = 0; element < matrix_size - diagonal; element++) {
            // Calculate the local set label
            long local_label = (diagonal + element) ^ element;

            // Skip if the local label does not match the set label
            if (x_label != local_label) continue;

            // Determine the sign for the set member
            long sign_set_member = (__builtin_popcount(element & z_member) & 1) ? -1 : 1;

            // Calculate the total sign
            long total_sign = sign_negative_i * sign_set_member;

            // Retrieve the values from the band matrices
            double value1 = B(diagonal, element);
            double value2 = B(-diagonal, element);

            // Update the weight
            weight += total_sign * (value1 + sign_reflection * value2);
        }
    }

    // Normalize the weight by the matrix size
    return weight / matrix_size;
}\end{lstlisting}

\end{document}